\pdfoutput=1
\pdfsuppressptexinfo=\numexpr 1+2+4+8+16+32+64+128+256+512 \relax

\documentclass[a4paper,oneside,fontsize=10pt,parskip=half,fleqn]{scrartcl}

\RequirePackage{scrhack}

\KOMAoptions{headings=standardclasses}

\RequirePackage[left={2cm},right={2cm},top={1.5cm},bottom={1.5cm},bindingoffset={0cm},includeheadfoot]{geometry}

\RequirePackage[utf8]{inputenc}
\RequirePackage[T1]{fontenc}

\RequirePackage{lmodern}
\RequirePackage[ngerman,english]{babel}

\RequirePackage{graphicx}

\RequirePackage[singlespacing]{setspace}

\RequirePackage[pdftoolbar=true,pdfmenubar=true,pdfstartview=FitV,pdfdisplaydoctitle=true,pdfencoding=auto,final]{hyperref}
\RequirePackage[open,openlevel=1,final]{bookmark}

\hypersetup{
	colorlinks			= true,
	allcolors			= black,
}

\RequirePackage{amsmath}
\RequirePackage{amsfonts}
\RequirePackage{amssymb}
\RequirePackage{amsthm}

\RequirePackage{abstract}

\RequirePackage[normal,nooneline,format=plain,singlelinecheck=off,labelfont={bf}]{caption}
\RequirePackage{floatrow}

\RequirePackage[final]{microtype}

\RequirePackage{csquotes}

\RequirePackage[backend=biber,style=alphabetic,uniquename=init,hyperref=true]{biblatex}

\AtEveryBibitem{%
	\iffieldundef{journaltitle}{}{\clearfield{eprint}}%
	\clearfield{issue}%
	\clearfield{number}%
}
\AtEveryCitekey{%
	\iffieldundef{journaltitle}{}{\clearfield{eprint}}%
	\clearfield{issue}%
	\clearfield{number}%
}

\DeclareFieldFormat[article]{volume}{\textbf{#1}\setunit*{\space}}
\DeclareFieldFormat[article]{pages}{#1}

\renewcommand*{\bibeidpunct}{\addcolon}

\renewbibmacro*{volume+number+eid}{%
	\printfield{volume}%
	\setunit*{\bibeidpunct}%
	\printfield{eid}%
	\setunit{\addcomma\space}%
}

\renewbibmacro{in:}{
	\ifentrytype{article}{}{\printtext{\bibstring{in}\intitlepunct}}%
}

\renewbibmacro*{journal}{%
	\iffieldundef{shortjournal}%
		{%
			\iffieldundef{journaltitle}%
			{}%
			{\printtext[journaltitle]%
				{%
					\printfield[titlecase]{journaltitle}%
					\setunit{\subtitlepunct}%
					\printfield[titlecase]{journalsubtitle}%
				}%
			}%
		}%
		{%
			\printtext[journaltitle]%
				{%
					\printfield[titlecase]{shortjournal}%
				}%
		}%
}

\DeclareCiteCommand{\fullcite}
{\usebibmacro{prenote}}
{\usedriver
	{\defcounter{maxnames}{12}}
	{\thefield{entrytype}}}
{\multicitedelim}
{\usebibmacro{postnote}}

\RequirePackage[noblocks]{authblk}

\RequirePackage{orcidlink}

\RequirePackage{xifthen}

\RequirePackage{xcolor}

\RequirePackage{enumitem}

\RequirePackage{helvet}

\RequirePackage{interval}

\title{Preserving Bifurcations through Moment Closures}

\author[1,2,3]{Christian Kuehn~\orcidlink{0000-0002-7063-6173}\,}
\author[1]{Jan Mölter~\orcidlink{0000-0002-5964-6207}\,}

\affil[1]{Department of Mathematics, School of Computation, Information and Technology, Technical University of Munich, Boltzmannstraße 3, 85748 Garching bei München, Germany}
\affil[2]{Munich Data Science Institute, Technical University of Munich, Walther-von-Dyck-Straße 10, 85748 Garching bei München, Germany}
\affil[3]{Complexity Science Hub Vienna, Josefstädter Straße 39, 1080 Vienna, Austria}

\date{}

\DeclareSymbolFont{dsrom}{U}{dsrom}{m}{n}
\DeclareMathSymbol{\dsI}{3}{dsrom}{"31}

\theoremstyle{plain}
\newtheorem{theorem}{Theorem}
\newtheorem{proposition}[theorem]{Proposition}
\newtheorem{lemma}[theorem]{Lemma}
\newtheorem{corollary}[theorem]{Corollary}

\newtheorem*{theorem*}{Theorem}
\newtheorem*{proposition*}{Proposition}
\newtheorem*{lemma*}{Lemma}
\newtheorem*{corollary*}{Corollary}

\theoremstyle{remark}

\newtheorem*{remark*}{Remark}

\theoremstyle{plain}
\newtheorem{innercustomtheorem}{Theorem}
\newenvironment{customtheorem}[1]{\renewcommand\theinnercustomtheorem{#1}\innercustomtheorem}{\endinnercustomtheorem}

\newcommand{\mathrelphantom}[1]{\ensuremath{\mathrel{\phantom{#1}}}}
\newcommand{\mathpunctuation}[1]{\ensuremath{\text{#1}}}

\newcommand{\slot}{\ensuremath{\, \mathbf{\cdot} \,}}

\newcommand{\naturals}{\ensuremath{\mathbb{N}}}
\newcommand{\reals}{\ensuremath{\mathbb{R}}}

\renewcommand{\Re}{\ensuremath{\operatorname{Re}}}

\DeclareMathOperator{\codim}{codim}

\DeclareMathOperator{\linearspan}{span}
\DeclareMathOperator{\kernel}{ker}
\DeclareMathOperator{\range}{ran}

\DeclareMathOperator{\rank}{rank}

\DeclareMathOperator{\trace}{tr}

\newcommand{\spectrum}[2][]{\ensuremath{\operatorname{spec_{#1}} #2}}

\newcommand{\inverse}[1]{\ensuremath{#1^{-1}}}
\newcommand{\transpose}[1]{\ensuremath{#1^{\top}}}
\newcommand{\orthogonal}[1]{\ensuremath{#1^{\perp}}}

\newcommand{\stack}[2]{\ensuremath{\genfrac{}{}{0pt}{}{#1}{#2}}}

\newcommand{\suchthat}{\ensuremath{\, : \,}}

\newcommand{\parenth}[1]{\ensuremath{\left( #1 \right)}}
\newcommand{\tparenth}[1]{\ensuremath{( #1 )}}

\newcommand{\bracket}[1]{\ensuremath{\left[ #1 \right]}}
\newcommand{\tbracket}[1]{\ensuremath{[ #1 ]}}

\newcommand{\anglebr}[1]{\ensuremath{\left\langle #1 \right\rangle}}
\newcommand{\tanglebr}[1]{\ensuremath{\langle #1 \rangle}}

\newcommand{\of}[1]{\ensuremath{\!\parenth{#1}}}
\newcommand{\tof}[1]{\ensuremath{\tparenth{#1}}}

\newcommand{\tabs}[1]{\ensuremath{\vert #1 \vert}}

\newcommand{\tnorm}[2][]{\ensuremath{\Vert #2 \Vert_{#1}}}
\newcommand{\inner}[3][]{\ensuremath{\anglebr{ #2 \vphantom{#3} , #3 \vphantom{#2}}_{#1}}}
\newcommand{\tinner}[3][]{\ensuremath{\tanglebr{ #2 \vphantom{#3} , #3 \vphantom{#2}}_{#1}}}

\newcommand{\set}[1]{\ensuremath{\left\lbrace #1 \right\rbrace}}
\newcommand{\tset}[1]{\ensuremath{\lbrace #1 \rbrace}}

\newcommand{\tball}[3][]{\ensuremath{B_{#2}^{#1}\tof{#3}}}

\newcommand{\varid}{\ensuremath{{\dsI}}}

\newcommand{\derivative}[2][]{\ensuremath{\left.\tfrac{\mathrm{d}}{\mathrm{d}#2}\ifthenelse{\isempty{#1}}{\right.}{\right\vert_{#2=#1}}}}

\DeclareMathOperator{\sign}{sgn}

\renewcommand{\O}{\ensuremath{\mathcal{O}}}

\newcommand{\tprojector}[2]{\ensuremath{\tinner{#1}{\cdot \,} #2}}

\newenvironment{psmallmatrix}{\left(\begin{smallmatrix}}{\end{smallmatrix}\right)}

\newcommand{\A}{\ensuremath{\mathrm{A}}}
\newcommand{\B}{\ensuremath{\mathrm{B}}}
\renewcommand{\S}{\ensuremath{\mathrm{S}}}
\newcommand{\I}{\ensuremath{\mathrm{I}}}
\newcommand{\X}{\ensuremath{\mathrm{X}}}

\makeatletter
\newcommand{\namedlabel}[2]{%
  \@bsphack
  \csname phantomsection\endcsname%
  \def\@currentlabel{#2}{\label{#1}}%
  \@esphack
}
\makeatother

\newcommand{\msclink}[1]{\href{https://zbmath.org/classification/?q=cc:#1}{#1}}

\begin{document}

\maketitle
\setcounter{page}{1}

\begin{abstract}
    Moment systems arise in a wide range of contexts and applications, e.g. in network modeling of complex systems. Since moment systems consist of a high or even infinite number of coupled equations, an indispensable step in obtaining a low-dimensional representation that is amenable to further analysis is, in many cases, to select a moment closure. A moment closure consists of a set of approximations that express certain higher-order moments in terms of lower-order ones, so that applying those leads to a closed system of equations for only the lower-order moments. Closures are frequently found drawing on intuition and heuristics to come up with quantitatively good approximations. In contrast to that, we propose an alternative approach where we instead focus on closures giving rise to certain qualitative features, such as bifurcations. Importantly, this fundamental change of perspective provides one with the possibility of classifying moment closures rigorously in regard to these features. This makes the design and selection of closures more algorithmic, precise, and reliable. In this work, we carefully study the moment systems that arise in the mean-field descriptions of two widely known network dynamical systems, the SIS epidemic and the adaptive voter model. We derive conditions that any moment closure has to satisfy so that the corresponding closed systems exhibit the transcritical bifurcation that one expects in these systems coming from the stochastic particle model.
    
    \vspace{2\parsep}
    
    
\end{abstract}

\setcounter{secnumdepth}{0}
\tableofcontents

\section{Introduction}

Moment systems are ubiquitous and appear in a wide range of applications, from statistical physics and kinetic theory~\parencite{bilger1993conditional,levermore1996moment,levermore1998gaussian,rangan2006maximum,abdelmalik2016moment} to chemical reaction systems and stochastic processes~\parencite{gillespie2009momentclosure,lee2009moment,bronstein2018variational,nasell2003extension,groenlund2013transcription,lombardo2014nonmonotonic} to contact processes on networks~\parencite{keeling1997correlation,eames2002modeling,gross2006epidemic,house2011insights,huepe2011adaptive,house2012modelling,taylor2012markovian,gleeson2013binary,demirel2014momentclosure,house2015algebraic,pellis2015exact,wuyts2022meanfield}. Moments are standard observables of large dynamical systems and deriving their evolution leads to coupled ordinary differential equations. In network dynamical systems, typical examples of moments are averages or densities computable from the current state. In general, i.e. when the system in question does not have a special structure, there is an infinite number of moments to consider that will then describe the system in its entirety. These high-dimensional systems are generally not amenable to further analysis.

The prevailing strategy in practice is therefore to reduce the size of these systems to one that is manageable either analytically or numerically by means of application of closure relations~\parencite{pellis2015exact,kuehn2016moment,wuyts2022meanfield}. Closure relations are generally approximations with which, from a certain order on, all higher-order moments are expressed in terms of the lower-order ones. This breaks the hierarchy of coupled ordinary differential equations and results in a closed system of equations for only the lower-order moments.

To illustrate this point, suppose we are given the system
\begin{equation}
	\begin{split}
		\dot{x}_{1} &= f_{1}\tof{x_{1}, x_{2}, \ldots} \\
		\dot{x}_{2} &= f_{2}\tof{x_{1}, x_{2}, \ldots} \\
		\vdots \hspace{1.5mm} &= \hspace{10mm} \vdots \\
	\end{split}
	\label{eq:general-moment-system}
\end{equation}
which consists of an infinite number of evolution equations for moments $x_{n}$, $n \in \naturals$, where we use Newton's notation for the derivative with respect to time. A closure relation for some $\kappa \in \naturals$ is a mapping $H$ such that $H\tof{x_{1}, x_{2}, \ldots x_{\kappa}} = \tparenth{x_{\kappa+1}, x_{\kappa+2}, \ldots}$ which renders the originally infinite-dimensional system the closed, finite-dimensional system
\begin{equation}
	\begin{split}
		\dot{x}_{1} &= f_{1}\tof{x_{1}, x_{2}, \ldots x_{\kappa}, H\tof{x_{1}, x_{2}, \ldots x_{\kappa}}} \\
		\dot{x}_{2} &= f_{2}\tof{x_{1}, x_{2}, \ldots x_{\kappa}, H\tof{x_{1}, x_{2}, \ldots x_{\kappa}}} \\
		\vdots \hspace{1.5mm} &= \hspace{10mm} \vdots \\
		\dot{x}_{\kappa} &= f_{\kappa}\tof{x_{1}, x_{2}, \ldots x_{\kappa}, H\tof{x_{1}, x_{2}, \ldots x_{\kappa}}} \\
	\end{split}
	\label{eq:general-closed-moment-system}
\end{equation}
Finding suitable closure relations involves a great deal of insight into the dynamics. Usually, they are motivated heuristically in various ways, depending on the system, and then justified empirically. Yet, the key question that has remained open is how well the closed system \eqref{eq:general-closed-moment-system} approximates the original system \eqref{eq:general-moment-system}. For general non-linear systems, this is an unsolved and extremely difficult problem. In part, this is because a given closure relation might not be suitable throughout the entire phase space of the moment system but only locally in some neighbourhood~\parencite{kuehn2016moment}. Therefore, trying to find globally approximating maps $H$ is extremely difficult. Even if one were to find an excellent approximating map $H$ by intuition or a closed black-box scheme, we would still be faced with the problem of explaining the structure of $H$.

Rather than trying to derive closures that approximate the original system quantitatively well across the entirety of phase space, we propose the alternative approach of focusing on closures that give rise to or preserve certain qualitative features of the system locally.
More specifically, here we take a dynamical perspective and characterise closure relations that qualitatively preserve bifurcations and, in particular, transcritical bifurcations. Bifurcations are an important dynamical feature corresponding to fundamental changes in the system under parameter variation~\parencite{guckenheimer1983nonlinear,crawford1991introduction}. As such, and for the low-dimensional closed system to have any meaningful relation to the original system dynamically, it is paramount to preserve them through the closure. Importantly, we argue that approaching the problem of finding closure relations from such a perspective enables one to rigorously derive conditions for their validity. This strategy provides an alternative, explainable, robust, and extensible route for determining closure relations by limiting the space of available closures by rigorous dynamical principles.

However, since even in a generic simple, i.e. linear, moment system this is still a very difficult problem (see below). In this work, we will consider two concrete, paradigmatic network dynamical systems, in particular, the mean-field moment systems associated with the SIS epidemic and the adaptive voter model, both of which we will introduce in what follows. The main results of this work then amount to rigorous conditions that moment closures for both systems must satisfy to give rise to the expected transcritical bifurcations. Finally, we explicitly confirm that the closures that are generally used do satisfy these conditions.

\section{Results}

\begin{figure}[tbp]
	\centering
	
	\includegraphics[page=1,trim={0 0.35cm 0 0},clip]{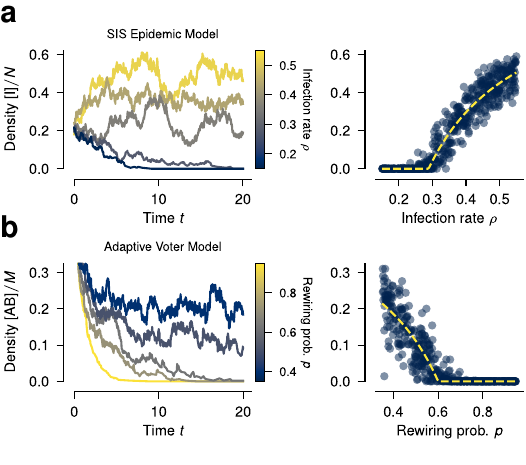}
	
	\caption{\textbf{Indications of a bifurcation in the SIS epidemic and the adaptive voter model.} Both the SIS epidemic (\textbf{\textsf{a}}) and the adaptive voter model (\textbf{\textsf{b}}) were simulated on a network with regular topology for different values of the infection rate and the rewiring probability, respectively. Depending on those, the prevalence of infected and active edges, respectively, either vanishes or remains positive after some finite time (left column). In fact, when considering this prevalence after some finite time in expectation, in both cases the system seems to exhibit a (supercritical) transcritical bifurcation (right column).}
	\label{fig:bifurcation-signatures}
\end{figure}

In a network dynamical system, every node of an underlying network is a dynamical system whose (in general, stochastic) dynamics is coupled to the dynamics of the nodes in its neighbourhood on the network~\parencite{newman2010networks,boccaletti2006complex,porter2016dynamical}. If the system is fully adaptive, then the topology of the network is also dynamic and coupled to the dynamics of the nodes~\parencite{gross2008adaptive,gross2009adaptive,sayama2013modeling}. As is the case for both the SIS epidemic and the adaptive voter model, we will assume that the state space of the nodes' dynamics is discrete. Such systems frequently admit a mean-field description in terms of network moments~\parencite{sharkey2008deterministic,house2009motif,sharkey2011deterministic,porter2016dynamical,kiss2017mathematics}. In this context, network moments are the expected number of certain motifs in the underlying network, such as a node in a certain state or a link connecting nodes in certain states. For instance, we will denote with $\tbracket{\sigma}$, $\tbracket{\sigma\,\sigma'}$, and $\tbracket{\sigma\,\sigma'\sigma''}$ the expected number of nodes in state $\sigma$, the expected number of pairs of nodes in states $\sigma$ and $\sigma'$ connected by a link, and the expected number of triples of nodes in states $\sigma$, $\sigma'$, and $\sigma''$ connected by links between the nodes in states $\sigma$ and $\sigma'$ as well as $\sigma'$ and $\sigma''$, respectively. Finally, we say that a network moment is of order $\kappa$ if its associated motif has size $\kappa$. Consequently, this enables us to classify network moments, so that, for instance, $\tbracket{\sigma}$, $\tbracket{\sigma\,\sigma'}$, and $\tbracket{\sigma\,\sigma'\sigma''}$ are moments of orders $1$, $2$, and $3$, respectively. Furthermore, we say that a moment system is of order $\kappa$ if it involves only evolution equations for moments up to order $\kappa$ and that a closure relation is of order $\kappa$ if its domain consists only of moments up to order $\kappa$.

For a network dynamical system with binary-state dynamics, such mean-field equations, for instance, of order $2$, then take the general form
\begin{equation}
	\left\lbrace
	\begin{split}
		\dot{\tbracket{\sigma_{1}}} &= f_{\tbracket{\sigma_{1}}}\tof{\tbracket{\sigma_{1}}, \tbracket{\sigma_{2}}, \tbracket{\sigma_{1}\sigma_{1}}, \tbracket{\sigma_{1}\sigma_{2}}, \tbracket{\sigma_{2}\sigma_{2}}, \ldots, \lambda} \\
		\dot{\tbracket{\sigma_{2}}} &= f_{\tbracket{\sigma_{2}}}\tof{\tbracket{\sigma_{1}}, \tbracket{\sigma_{2}}, \tbracket{\sigma_{1}\sigma_{1}}, \tbracket{\sigma_{1}\sigma_{2}}, \tbracket{\sigma_{2}\sigma_{2}}, \ldots, \lambda} \\
		\dot{\tbracket{\sigma_{1}\sigma_{1}}} &= f_{\tbracket{\sigma_{1}\sigma_{1}}}\tof{\tbracket{\sigma_{1}}, \tbracket{\sigma_{2}}, \tbracket{\sigma_{1}\sigma_{1}}, \tbracket{\sigma_{1}\sigma_{2}}, \tbracket{\sigma_{2}\sigma_{2}}, \ldots, \lambda} \\
		\dot{\tbracket{\sigma_{1}\sigma_{2}}} &= f_{\tbracket{\sigma_{1}\sigma_{2}}}\tof{\tbracket{\sigma_{1}}, \tbracket{\sigma_{2}}, \tbracket{\sigma_{1}\sigma_{1}}, \tbracket{\sigma_{1}\sigma_{2}}, \tbracket{\sigma_{2}\sigma_{2}}, \ldots, \lambda} \\
		\dot{\tbracket{\sigma_{2}\sigma_{2}}} &= f_{\tbracket{\sigma_{2}\sigma_{2}}}\tof{\tbracket{\sigma_{1}}, \tbracket{\sigma_{2}}, \tbracket{\sigma_{1}\sigma_{1}}, \tbracket{\sigma_{1}\sigma_{2}}, \tbracket{\sigma_{2}\sigma_{2}}, \ldots, \lambda} \\
	\end{split}
	\right.
	\label{eq:general-network-moment-system}
\end{equation}
where $f_{M}$ for some moment $M$ denotes the component of a vector field $f$ in the \enquote{direction} $M$, $\lambda$ is a parameter, and the ellipses indicate potential dependence on moments of order $3$ and higher.

As we will see in the following, the mean-field moment systems of the SIS epidemic and the adaptive voter model are precisely of this form.

Moreover, corresponding to the expected number of certain network motifs, the phase space variables are non-negative quantities. Hence only the non-negative orthant of the phase space is physically relevant, and in such a scenario, we can distinguish between two types of transcritical bifurcations. A transcritical bifurcation is called supercritical if we follow a stable branch in parameter space and, above the bifurcation point there exists another stable branch; it is called subcritical otherwise~\parencite{kuehn2021universal}. The distinction between super- and subcritical scenarios is particularly important for applications, as supercritical changes only induce a small change in the state, while subcritical ones can lead to drastic large jumps in the state under parametric variation~\parencite{gross2006epidemic,iacopini2019simplicial}.

Finally, we note that mean-field moment systems of contact processes on networks are in many cases linear, and it is only through the application of a closure relation that non-linear terms are introduced in the system.
However, this observation does not necessarily simplify the problem of deriving conditions on the closure relations that give rise to a certain bifurcation.

To illustrate this, consider the system $\dot{\Xi} = F\tof{\Xi, \lambda} := A_{\lambda}\tbracket{\Xi} + R_{\lambda}\tbracket{H\tof{\Xi}}$ with $\Xi \in \reals^{d}$, $A_{\lambda}: \reals^{d} \to \reals^{d}$ and $R_{\lambda}: \reals^{m} \to \reals^{d}$ linear operators, and $H: \reals^{d} \to \reals^{m}$ a smooth, non-linear map with $H\tof{0} = 0$.
A necessary condition for the presence of a transcritical bifurcation is that there exists $\lambda_{*}$ such that $\rank{D_{\Xi}F\tof{0, \lambda_{*}}} = \rank\tof{A_{\lambda_{*}} + R_{\lambda_{*}} D_{\Xi}H\of{0}} = d - 1$.
Yet, even under strong generic assumptions on $A_{\lambda}$ and $R_{\lambda}$ this is hard to guarantee. This is because the rank is subadditive, so for matrices $A$ and $R$ we only have that $\tabs{\rank{A} - \rank{R}} \leq \rank\tof{A + R} \leq \rank{A} + \rank{R}$~\parencite{horn2012matrix}. But this provides very little control, and one can hardly derive any concrete conditions on $H$.

\subsection{The SIS epidemic model}

The SIS epidemic model describes the spread of an epidemic in a population without immunity so that individuals are either susceptible (\enquote{S}) or infected (\enquote{I}) and upon infection or recovery switch from one state to the other. The infection spreads via the links of an underlying network, and, subject to Poisson processes, susceptible individuals become infected from connected infected individuals, or infected individuals recover.

In the mean field, these dynamics can be described by the moment system~\parencite{kiss2017mathematics,simon2011exact}
\begin{equation}
	\left\lbrace
	\begin{split}
		\dot{\tbracket{\S}} &= \tbracket{\I} - \rho \, \tbracket{\S\I} \\
		\dot{\tbracket{\I}} &= \rho \, \tbracket{\S\I} - \tbracket{\I} \\
		\dot{\tbracket{\S\S}} &= 2 \tbracket{\S\I} - 2 \rho \tbracket{\S\S\I} \\
		\dot{\tbracket{\S\I}} &= \tbracket{\I\I} - \tbracket{\S\I} + \rho \tparenth{\tbracket{\S\S\I} - \tbracket{\I\S\I} - \tbracket{\S\I}} \\
		\dot{\tbracket{\I\I}} &= -2 \tbracket{\I\I} + 2 \rho \tparenth{\tbracket{\I\S\I} + \tbracket{\S\I}} \\
	\end{split}
	\right.
	\label{eq:general-sis-epidemic-model}
\end{equation}
where the infection rate $\rho$ is a strictly positive parameter, and time is measured in units such that the recovery rate is $1$.
For this system, one notes that $\dot{\tbracket{\S}} + \dot{\tbracket{\I}} = 0$ and $\dot{\tbracket{\S\S}} + 2 \dot{\tbracket{\S\I}} + \dot{\tbracket{\I\I}} = 0$. Consequently, the total number of individuals $N$ and links $M$ are conserved under the dynamics, so that the above system reduces to the equivalent system
\begin{equation}
	\begin{split}
		&\left\lbrace
		\begin{split}
			\dot{\tbracket{\I}} &= \rho \, \tbracket{\S\I} - \tbracket{\I} \\
			\dot{\tbracket{\S\I}} &= \tbracket{\I\I} - \tbracket{\S\I} + \rho \tparenth{\tbracket{\S\S\I} - \tbracket{\I\S\I} - \tbracket{\S\I}} \\
			\dot{\tbracket{\I\I}} &= -2 \tbracket{\I\I} + 2 \rho \tparenth{\tbracket{\I\S\I} + \tbracket{\S\I}} \\
		\end{split}
		\right.
		\\[0.5cm]
		&\quad \text{\parbox{\linewidth-1.5cm}{subject to $\tbracket{\S} + \tbracket{\I} = N$ and $\tbracket{\S\S} + 2 \tbracket{\S\I} + \tbracket{\I\I} = 2 M$.}}
	\end{split}
	\label{eq:reduced-general-sis-epidemic-model}
\end{equation}

Upon variation of the infection rate $\rho$, we expect the SIS epidemic model to exhibit a transcritical bifurcation in, e.g. the dynamic variable $\tbracket{\I}$, the number of infected individuals (Fig.~\ref{fig:bifurcation-signatures}a). Below the bifurcation point, the so-called disease-free equilibrium with a vanishing number of infected individuals ($\tbracket{\I} = 0$) is stable, whereas, above that point, the epidemic spread can sustain itself, and the system tends to a dynamic equilibrium with a non-vanishing number of infected individuals ($\tbracket{\I} > 0$).

\namedlabel{thms:sis-epidemic-model-closure}{I}

Now, for a closure relation $H$ of order $1$ with $\tbracket{\S\I} = H\tof{\tbracket{\I}}$ we have the following result.

\begin{customtheorem}{I.1}
	\label{thm:sis-epidemic-model-order-1-closure}
	Assume that $H$ is at least twice continuously differentiable in a neighbourhood around $0$ and that $H\tof{0} = 0$, and suppose that $\frac{1}{\rho_{*}} = H'\tof{0} > 0$.
	Then, if $H''\tof{0} \neq 0$, the closed system exhibits a transcritical bifurcation at $\rho = \rho_{*}$.
	
	In particular, the bifurcation is supercritical (subcritical) if $H''\tof{0} \overset{\tof{>}}{<} 0$.
\end{customtheorem}

Despite appearing abstract at first sight, the last result can be interpreted quite directly. In fact, under the assumptions of the theorem, we can expand $H$ around $0$ in a Taylor series and obtain that $H\tof{\tbracket{\I}} = H'\tof{0} \tbracket{\I} + \frac{1}{2} H''\tof{0} \tbracket{\I}^{2} + \O\tof{\tbracket{\I}^{3}}$. Thus,
\begin{equation*}
	\dot{\tbracket{\I}} = \rho \, H\tof{\tbracket{\I}} - \tbracket{\I} = \tparenth{\rho H'\tof{0} - 1} \tbracket{\I} + \frac{1}{2} H''\tof{0} \tbracket{\I}^{2} + \O\tof{\tbracket{\I}^{3}} \mathpunctuation{,}
\end{equation*}
which one recognises as essentially the normal form of a transcritical bifurcation if $H''\tof{0} \neq 0$. Moreover, it is a supercritical bifurcation if $H''\tof{0} < 0$ and a subcritical otherwise.

If, instead of a closure relation of order $1$, we consider a closure relation $H$ of order $2$ with $\tparenth{\tbracket{\S\S\I}, \tbracket{\I\S\I}} = H\tof{\tbracket{\I}, \tbracket{\S\I}, \tbracket{\I\I}} =: \tparenth{H^{\tparenth{\tbracket{\S\S\I}}}, H^{\tparenth{\tbracket{\I\S\I}}}}\tof{\tbracket{\I}, \tbracket{\S\I}, \tbracket{\I\I}}$, we get the following.

\begin{customtheorem}{I.2}
	\label{thm:sis-epidemic-model-order-2-closure}
	Assume that $H$ is rational and that $H\tof{\tbracket{\I}, \tbracket{\S\I}, \tbracket{\I\I}} = 0$ whenever $\tbracket{\S\I} = 0$. Then $H$ can be factorised so that $H\tof{\tbracket{\I}, \tbracket{\S\I}, \tbracket{\I\I}} = \tbracket{\S\I} \, \tilde{H}\tof{\tbracket{\I}, \tbracket{\S\I}, \tbracket{\I\I}}$.
	Moreover, suppose that $\tilde{H}$ is at least twice continuously differentiable in a neighbourhood around $0$ and that $\frac{1}{\rho_{*}} = \tilde{H}^{\tparenth{\tbracket{\S\S\I}}}\tof{0} > 0$.
	Then, if
	\begin{equation*}
		\partial_{\tbracket{\I}} \tilde{H}^{\tparenth{\tbracket{\S\S\I}}}\tof{0} + \tilde{H}^{\tparenth{\tbracket{\S\S\I}}}\tof{0} \partial_{\tbracket{\S\I}} \tilde{H}^{\tparenth{\tbracket{\S\S\I}}}\tof{0} + \tparenth{1 + \tilde{H}^{\tparenth{\tbracket{\I\S\I}}}\tof{0}} \partial_{\tbracket{\I\I}} \tilde{H}^{\tparenth{\tbracket{\S\S\I}}}\tof{0} \neq 0 \mathpunctuation{,}
	\end{equation*}
	the closed system exhibits a transcritical bifurcation at $\rho = \rho_{*}$.
	
	In particular, provided that $2 \tilde{H}^{\tparenth{\tbracket{\S\S\I}}}\tof{0} + \tilde{H}^{\tparenth{\tbracket{\I\S\I}}}\tof{0} + 1 > 0$, the bifurcation is supercritical (subcritical) if
	\begin{equation*}
		\partial_{\tbracket{\I}} \tilde{H}^{\tparenth{\tbracket{\S\S\I}}}\tof{0} + \tilde{H}^{\tparenth{\tbracket{\S\S\I}}}\tof{0} \partial_{\tbracket{\S\I}} \tilde{H}^{\tparenth{\tbracket{\S\S\I}}}\tof{0} + \tparenth{1 + \tilde{H}^{\tparenth{\tbracket{\I\S\I}}}\tof{0}} \partial_{\tbracket{\I\I}} \tilde{H}^{\tparenth{\tbracket{\S\S\I}}}\tof{0} \overset{\tof{>}}{<} 0 \mathpunctuation{.}
	\end{equation*}
\end{customtheorem}

\begin{remark*}
The condition that the closure vanishes whenever $\tbracket{\S\I} = 0$ is to ensure consistency and motivated through the network topology. Then, the condition $\tilde{H}^{\tparenth{\tbracket{\S\S\I}}}\tof{0} > 0$ ensures the existence of the bifurcation and the condition $2 \tilde{H}^{\tparenth{\tbracket{\S\S\I}}}\tof{0} + \tilde{H}^{\tparenth{\tbracket{\I\S\I}}}\tof{0} + 1 > 0$, the exchange of stability at the bifurcation, which one finds to be frequently satisfied with a closure approximating network motifs.
\end{remark*}

Proving the existence of bifurcations in higher dimensions tends to be a bit more involved. As such, we defer the proof of this result to the next section. However, before moving on, let's briefly pause and reconsider some closure relations that are used in the context of the SIS epidemic model in light of Theorems~\ref{thm:sis-epidemic-model-order-1-closure} and \ref{thm:sis-epidemic-model-order-2-closure}.

\begin{figure}[tbp]
	\centering
	
	\includegraphics[page=2,trim={0 0.45cm 0 0},clip]{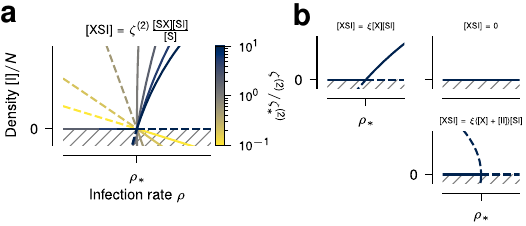}
	
	\caption{\textbf{Bifurcation diagrams for the SIS epidemic model under different closures.} Depending on the closure, the SIS epidemic model can exhibit very different bifurcation behaviours. Under the frequently used closure $\tbracket{\X\S\I} = \zeta^{\tparenth{2}} \frac{\tbracket{\S\X} \tbracket{\S\I}}{\tbracket{\S}}$ ($\zeta^{\tparenth{2}} > 0$) (\textbf{\textsf{a}}), and depending on whether $\zeta^{\tparenth{2}} > \zeta^{\tparenth{2}}_{*}$ or $\zeta^{\tparenth{2}} < \zeta^{\tparenth{2}}_{*}$ for $\zeta^{\tparenth{2}}_{*} = \frac{1}{2} \tparenth{1 - \frac{N}{2 M}}$, we can observe a super- or subcritical transcritical bifurcation, respectively. Alternative closures (\textbf{\textsf{b}}) can give rise to the same, a different, or no bifurcation. For instance, the closure $\tbracket{\X\S\I} = \xi \tbracket{\X} \tbracket{\S\I}$ (top left) gives rise to a supercritical transcritical bifurcation, the closure $\tbracket{\X\S\I} = \xi \tparenth{\tbracket{\X} + \tbracket{\I\I}}$ (bottom right) to a subcritical pitchfork bifurcation, and the closure $\tbracket{\X\S\I} = 0$ (top right) to no bifurcation.}
	\label{fig:SIS-bifurcations}
\end{figure}

At order $1$, the most frequently used closure relation is the homogeneous closure $\tbracket{\S\I} = \zeta^{\tparenth{1}} \tbracket{\S} \tbracket{\I}$ for a strictly positive constant $\zeta^{\tparenth{1}}$ (typically, $\zeta^{\tparenth{1}} = \frac{\tanglebr{k}}{N}$ with $\tanglebr{k}$ the mean degree)~\parencite{simon2011exact}. For this, $H\tof{\tbracket{\I}} = \zeta^{\tparenth{1}} \tparenth{N - \tbracket{\I}} \tbracket{\I}$ with $H'\tof{0} = \zeta^{\tparenth{1}} N$ and $H''\tof{0} = - 2 \zeta^{\tparenth{1}}$, which indeed yields a supercritical transcritical bifurcation at $\rho_{*} = \frac{1}{\zeta^{\tparenth{1}} N}$.

Similarly, at order $2$, the most frequently used closure relation is $\tbracket{\S\S\I} = \zeta^{\tparenth{2}} \frac{\tbracket{\S\S} \tbracket{\S\I}}{\tbracket{\S}}$ and $\tbracket{\I\S\I} = \zeta^{\tparenth{2}} \frac{\tbracket{\S\I}^{2}}{\tbracket{\S}}$ for a strictly positive constant $\zeta^{\tparenth{2}}$ (typically, if the topology is approximately regular, $\zeta^{\tparenth{2}} = \frac{\tanglebr{k} - 1}{\tanglebr{k}}$ with $\tanglebr{k}$ the mean degree, while, if the precise topology is unknown, $\zeta^{\tparenth{2}} = 1$ is a reasonable assumption coming from an Erdős–Rényi topology)~\parencite{sharkey2008deterministic,sharkey2011deterministic,simon2011exact}. Thus, $H\tof{\tbracket{\I}, \tbracket{\S\I}, \tbracket{\I\I}} = \frac{\zeta^{\tparenth{2}}}{N - \tbracket{\I}} \tparenth{\tparenth{2 M - 2 \tbracket{\S\I} - \tbracket{\I\I}} \tbracket{\S\I}, \tbracket{\S\I}^{2}}$ and $\tilde{H}\tof{\tbracket{\I}, \tbracket{\S\I}, \tbracket{\I\I}} = \frac{\zeta^{\tparenth{2}}}{N - \tbracket{\I}} \tparenth{2 M - 2 \tbracket{\S\I} - \tbracket{\I\I}, \tbracket{\S\I}}$ so that $\tilde{H}^{\tparenth{\tbracket{\S\S\I}}}\tof{0} = \frac{2 \zeta^{\tparenth{2}} M}{N}$ and $\tilde{H}^{\tparenth{\tbracket{\I\S\I}}}\tof{0} = 0$, and therefore $2 \tilde{H}^{\tparenth{\tbracket{\S\S\I}}}\tof{0} + \tilde{H}^{\tparenth{\tbracket{\I\S\I}}}\tof{0} + 1 > 0$ and $\partial_{\tbracket{\I}} \tilde{H}^{\tparenth{\tbracket{\S\S\I}}}\tof{0} + \tilde{H}^{\tparenth{\tbracket{\S\S\I}}}\tof{0} \partial_{\tbracket{\S\I}} \tilde{H}^{\tparenth{\tbracket{\S\S\I}}}\tof{0} + \tparenth{1 + \tilde{H}^{\tparenth{\tbracket{\I\S\I}}}\tof{0}} \partial_{\tbracket{\I\I}} \tilde{H}^{\tparenth{\tbracket{\S\S\I}}}\tof{0} = - \frac{4 \zeta^{\tparenth{2}} M}{N^{2}} \tparenth{\zeta^{\tparenth{2}} - \frac{1}{2} \tparenth{1 - \frac{N}{2 M}}}$. Thus, if $\zeta^{\tparenth{2}} > \frac{1}{2}\tparenth{1 - \frac{N}{2 M}}$, this yields a supercritical transcritical bifurcation at $\rho_{*} = \frac{N}{2 \zeta^{\tparenth{2}} M}$ and otherwise a subcritical transcritical bifurcation (Fig.~\ref{fig:SIS-bifurcations}a).

This closure relation has proved to be effective also quantitatively when the network topology is (close to) regular. To account for clustering in the underlying network topology, the closures $\tbracket{\X\S\I} = \zeta^{\tparenth{2}} \frac{\tbracket{\X\S} \tbracket{\S\I}}{\tbracket{\S}} \tparenth{1 - \phi \tparenth{1 - \frac{1}{\zeta^{\tparenth{1}}} \frac{\tbracket{\X\I}}{\tbracket{\X} \tbracket{\I}}}}$ or $\tbracket{\X\S\I} = \zeta^{\tparenth{2}} \frac{\tbracket{\X\S} \tbracket{\S\I}}{\tbracket{\S}} \tparenth{1 - \phi \tparenth{1 - \zeta^{\tparenth{1}} N \frac{\tbracket{\S}^{2} \tbracket{\I} \tbracket{\X\I}}{\tparenth{\tbracket{\S\S} \tbracket{\I} + \tbracket{\S} \tbracket{\I\I}} \tbracket{\X} \tbracket{\S\I}}}}$ for $\X \in \tset{\S, \I}$ with clustering coefficient $\phi$ have been proposed~\parencite{keeling1999effects,rand1999correlation,sharkey2006pairlevel,house2010impact}.
Yet, Theorem~\ref{thm:sis-epidemic-model-order-2-closure} does not immediately apply to these closures, since $\frac{\tbracket{\X\S\I}}{\tbracket{\S\I}}$ is not continuous and thus also not differentiable at $0$. Instead, one might consider the corresponding regularised closures $\tbracket{\X\S\I}_{\delta} = \zeta^{\tparenth{2}} \frac{\tbracket{\X\S} \tbracket{\S\I}}{\tbracket{\S}} \tparenth{1 - \phi \tparenth{1 - \frac{1}{\zeta^{\tparenth{1}}} \frac{\tbracket{\X\I}}{\tbracket{\X} \tbracket{\I} + \delta}}}$ and $\tbracket{\X\S\I}_{\delta} = \zeta^{\tparenth{2}} \frac{\tbracket{\X\S} \tbracket{\S\I}}{\tbracket{\S}} \tparenth{1 - \phi \tparenth{1 - \zeta^{\tparenth{1}} N \frac{\tbracket{\S}^{2} \tbracket{\I} \tbracket{\X\I}}{\tparenth{\tbracket{\S\S} \tbracket{\I} + \tbracket{\S} \tbracket{\I\I}} \tbracket{\X} \tbracket{\S\I} + \delta}}}$ for some regularisation parameter $0 < \delta \ll 1$, which ensures that the singularity in these closures is outside the physically relevant region of the phase space. In both cases and provided that $\phi \neq 1$, one can show that there occurs a transcritical bifurcation at $\rho_{*} = \frac{N}{2 \tparenth{1 - \phi} \zeta^{\tparenth{2}} M}$. Moreover, for the first of these two closures, the bifurcation is always subcritical if $\delta$ is sufficiently small, while for the second the type of bifurcation is independent of $\delta$ and supercritical if $\tparenth{1 - \phi} \zeta^{\tparenth{2}} > \frac{1}{2}\tparenth{1 - \frac{N}{2 M}}$ and subcritical otherwise.

Finally, the closure $\tbracket{\X\S\I} = \frac{\tbracket{\X\S} \tbracket{\S\I}}{\tbracket{\S\S} + \tbracket{\S\I}} \tparenth{\frac{\tanglebr{k^{2}} \tparenth{\tanglebr{k^{2}} \tbracket{\S} - \tanglebr{k} \tparenth{\tbracket{\S\S} + \tbracket{\S\I}}} + \tanglebr{k^{3}} \tparenth{\tparenth{\tbracket{\S\S} + \tbracket{\S\I}} - \tanglebr{k} \tbracket{\S}}}{\tparenth{\tbracket{\S\S} + \tbracket{\S\I}} \tparenth{\tanglebr{k^{2}} - \tanglebr{k}^{2}}} - 1}$ for $\X \in \tset{\S, \I}$ with $\tanglebr{k^{m}}$ the $m$\textsuperscript{th} moment of the degree-distribution has been proposed specifically for more irregular network topologies~\parencite{simon2016super}. Again, one can show that there occurs a transcritical bifurcation at $\rho_{*} = \frac{1}{\frac{\tanglebr{k^{2}}}{\tanglebr{k}} - 1}$, which is supercritical if $\tanglebr{k} - 2 \tanglebr{k^{2}} + \tanglebr{k^{3}} = \tanglebr{k \tparenth{k-1}^{2}} > 0$ and subcritical otherwise.

Besides these established closures, we can also imagine other closures. One alternative could be $\tbracket{\X\S\I} = \xi \tbracket{\X} \tbracket{\S\I}$ for $\X \in \tset{\S, \I}$ and a strictly positive constant $\xi$. In fact, one can verify that this also yields a supercritical transcritical bifurcation, in this case, at $\rho_{*} = \frac{1}{\xi N}$ (Fig.~\ref{fig:SIS-bifurcations}b).
In contrast to that, if we consider the pure and direct truncation closure $\tbracket{\X\S\I} = 0$ which neglects all higher-order moments, the system becomes linear, and one finds that it does not satisfy the conditions above, and it is also immediately obvious from the phase diagram that we do not get a transcritical bifurcation in this case (Fig.~\ref{fig:SIS-bifurcations}b). Similarly, the closure $\tbracket{\X\S\I} = \xi \tbracket{\S\I} \tparenth{\tbracket{\X} + \tbracket{\I\I}}$ fails to yield a transcritical bifurcation (Fig.~\ref{fig:SIS-bifurcations}b). While the latter is explicitly constructed as a counterexample, it is structurally rather similar to the alternative closure from above. Yet, it demonstrates that the conditions we have derived allow us to effectively draw a distinction between them. 

\subsection{The adaptive voter model}

The adaptive voter model is an extension of classical models of opinion and consensus formation~\parencite{clifford1973model,holley1975ergodic} concerned with a population in which every individual subscribes to either one of two mutually exclusive, opposing opinions, $\A$ and $\B$~\parencite{holme2006nonequilibrium,demirel2014momentclosure,zschaler2012adaptive}. Individuals interact via the links of an underlying network, and, subject to a Poisson process, when there are two connected individuals with different opinions, either one individual adopts the opinion of the other with probability $1 - p$, or the link between them is rewired from one to another individual with the same opinion with probability $p$.

In the mean field, these dynamics can be described by the moment system~\parencite{zschaler2012adaptive,demirel2014momentclosure}
\begin{equation}
	\left\lbrace
	\begin{split}
		\dot{\tbracket{\A}} &= 0 \\
		\dot{\tbracket{\B}} &= 0 \\
		\dot{\tbracket{\A\A}} &= \frac{1}{2} \tbracket{\A\B} + \frac{1-p}{2} \tparenth{2\,\tbracket{\A\B\A} - \tbracket{\A\A\B}} \\
		\dot{\tbracket{\A\B}} &= -\tbracket{\A\B} - \frac{1-p}{2} \tparenth{2\,\tbracket{\A\B\A} - \tbracket{\A\A\B} - \tbracket{\A\B\B} + 2\,\tbracket{\B\A\B}} \\
		\dot{\tbracket{\B\B}} &= \frac{1}{2} \tbracket{\A\B} + \frac{1-p}{2} \tparenth{2\,\tbracket{\B\A\B} - \tbracket{\A\B\B}} \\
	\end{split}
	\right.
	\label{eq:general-adaptive-voter-model}
\end{equation}
where the rewiring probability $p$ is strictly contained in the unit interval and time is measured in units such that the activation rate is $1$.
Similarly as in the mean-field moment system of the SIS epidemic model, one notes that $\dot{\tbracket{\A\A}} + \dot{\tbracket{\A\B}} + \dot{\tbracket{\B\B}} = 0$ so that the total number of links $M$ is conserved under the dynamics. Moreover, $\dot{\tbracket{\A}} = 0$ and $\dot{\tbracket{\B}} = 0$ so that not only the total number of individuals but also the numbers of individuals with opinions $\A$ and $\B$ are conserved. Hence, the system reduces to the equivalent system
\begin{equation}
	\begin{split}
		&\left\lbrace
		\begin{split}
			\dot{\tbracket{\A\A}} &= \frac{1}{2} \tbracket{\A\B} + \frac{1-p}{2} \tparenth{2\,\tbracket{\A\B\A} - \tbracket{\A\A\B}} \\
			\dot{\tbracket{\B\B}} &= \frac{1}{2} \tbracket{\A\B} + \frac{1-p}{2} \tparenth{2\,\tbracket{\B\A\B} - \tbracket{\A\B\B}} \\
		\end{split}
		\right.
		\\[0.5cm]
		&\quad \text{\parbox{\linewidth-1.5cm}{subject to $\tbracket{\A} = \tbracket{\A}_{0}$, $\tbracket{\B} = \tbracket{\B}_{0}$ and $\tbracket{\A\A} + \tbracket{\A\B} + \tbracket{\B\B} = M$.}}
	\end{split}
	\label{eq:reduced-general-adaptive-voter-model}
\end{equation}

Upon variation of the rewiring probability $p$, the adaptive voter model exhibits a transcritical bifurcation in the dynamic variable $\tbracket{\A\B} = M - \tbracket{\A\A} - \tbracket{\B\B}$, corresponding to the number of \enquote{active} links (Fig.~\ref{fig:bifurcation-signatures}b). Whereas below the bifurcation point, the system tends to a dynamic equilibrium with a non-vanishing number of active links ($\tbracket{\A\B} > 0$), above that point, the equilibrium with a vanishing number of active links ($\tbracket{\A\B} = 0$) becomes stable. To reach this state, the network breaks apart into several clusters, and in each individual cluster consensus is formed. Overall, only non-active links remain. However, their composition is undetermined, so that in phase space this equilibrium constitutes a manifold. Crucially, this entails the equilibrium being degenerate so that standard bifurcation results are no longer applicable.

For the adaptive voter model, a moment closure of order $2$ is a mapping $H$ with $\tparenth{\tbracket{\A\A\B}, \tbracket{\A\B\A}, \tbracket{\B\A\B}, \tbracket{\A\B\B}} = H\tof{\tbracket{\A\A}, \tbracket{\B\B}}$. Alternatively, since the system effectively only depends on differences of moments of order $3$, we may simply write it as $\tparenth{2\,\tbracket{\A\B\A} - \tbracket{\A\A\B}, 2\,\tbracket{\B\A\B} - \tbracket{\A\B\B}} = H\tof{\tbracket{\A\A}, \tbracket{\B\B}} =: \tparenth{H^{\tparenth{2\,\tbracket{\A\B\A} - \tbracket{\A\A\B}}}, H^{\tparenth{2\,\tbracket{\B\A\B} - \tbracket{\A\B\B}}}}\linebreak\tof{\tbracket{\A\A}, \tbracket{\B\B}}$.
Then for such a closure relation, we have the following result.

\begin{customtheorem}{II}
	\label{thm:adaptive-voter-model-order-2-closure}
	Assume that $H$ is rational and that $H\tof{\tbracket{\A\A}, \tbracket{\B\B}} = 0$ whenever $\tbracket{\A\A} + \tbracket{\B\B} = M$. Then $H$ can be factorised so that $H\tof{\tbracket{\A\A}, \tbracket{\B\B}} = \tparenth{M - \tbracket{\A\A} - \tbracket{\B\B}} \, \tilde{H}\tof{\tbracket{\A\A}, \tbracket{\B\B}}$.
	
	Moreover, suppose that there exists $\theta_{*} \in \interval[open]{0}{1}$ such that
	\begin{equation*}
		-\frac{1}{1 - p_{*}} = \tilde{H}^{\tparenth{2\,\tbracket{\A\B\A} - \tbracket{\A\A\B}}}\tof{\theta_{*} M, \tparenth{1-\theta_{*}} M} = \tilde{H}^{\tparenth{2\,\tbracket{\B\A\B} - \tbracket{\A\B\B}}}\tof{\theta_{*} M, \tparenth{1-\theta_{*}} M} < -1
	\end{equation*}
	and that $\tilde{H}$ is continuously differentiable in a neighbourhood around the point $\tparenth{\theta_{*} M, \tparenth{1-\theta_{*}} M}$. Then, if
	\begin{equation*}
		\det{D_{\tparenth{\tbracket{\A\A}, \tbracket{\B\B}}} \tilde{H}\tof{\theta_{*} M, \tparenth{1-\theta_{*}} M}} \neq 0
	\end{equation*}
	and
	\begin{equation*}
		\trace{\begin{psmallmatrix} +1 & -1 \\ -1 & +1 \end{psmallmatrix} D_{\tparenth{\tbracket{\A\A}, \tbracket{\B\B}}} \tilde{H}\tof{\theta_{*} M, \tparenth{1-\theta_{*}} M}} \neq 0 \mathpunctuation{,}
	\end{equation*}
	the closed system exhibits a degenerate transcritical bifurcation at $p = p_{*}$ and at this point $\tparenth{\tbracket{\A\A}, \tbracket{\B\B}} = \tparenth{\theta_{*} M, \tparenth{1-\theta_{*}} M}$.
	
	In particular, provided that
	\begin{equation*}
		\det{ D_{\tparenth{\tbracket{\A\A}, \tbracket{\B\B}}} \tilde{H}\tof{\theta_{*} M, \tparenth{1-\theta_{*}} M} } > 0
	\end{equation*}
	and
	\begin{equation*}
		\trace{ D_{\tparenth{\tbracket{\A\A}, \tbracket{\B\B}}} \tilde{H}\tof{\theta_{*} M, \tparenth{1-\theta_{*}} M} } \trace{ \begin{psmallmatrix} +1 & -1 \\ -1 & +1 \end{psmallmatrix} D_{\tparenth{\tbracket{\A\A}, \tbracket{\B\B}}} \tilde{H}\tof{\theta_{*} M, \tparenth{1-\theta_{*}} M} } > 0 \mathpunctuation{,}
	\end{equation*}
	there is an exchange of stability at the bifurcation, and the bifurcation itself is supercritical (subcritical) if
	\begin{equation*}
		\trace{D_{\tparenth{\tbracket{\A\A}, \tbracket{\B\B}}} \tilde{H}\tof{\theta_{*} M, \tparenth{1-\theta_{*}} M}} \overset{\tof{>}}{<} 0 \mathpunctuation{.}
	\end{equation*}
\end{customtheorem}

\begin{remark*}
The condition that the closure vanishes whenever $\tbracket{\A\B} = 0$ is to ensure consistency and motivated through the network topology. Moreover, the condition that both components of the closure coincide at a point on the trivial manifold $\tbracket{\A\A} + \tbracket{\B\B} = M$ ensures the existence of the bifurcation.
\end{remark*}

As before, we defer the proof of this to the next section. In comparison to the bifurcation in the SIS epidemic model, the transcritical bifurcation scenario in the adaptive voter model is more complicated due to the degeneracy of the trivial equilibrium. Yet, the overall strategy still works, showing that it can be effectively adapted. Furthermore, it is worthwhile to note here that we have not made any assumptions about symmetries of the closure relation. However, the adaptive voter model is symmetric with respect to the exchange of $\A$ and $\B$. While the symmetry does not pertain to preserving the bifurcation, in practice, ideally, one might want a closure relation to also preserve it. A constraint ensuring this would then further restrict the set of admissible closure relations.

We can now apply Theorem~\ref{thm:adaptive-voter-model-order-2-closure} to the closure relations that have been previously proposed and used in the context of the adaptive voter model. The most prominent one consists of the relations $\tbracket{\A\A\B} = 2 \frac{\kappa_{\A}}{\tbracket{\A}} \tbracket{\A\A} \tbracket{\A\B}$, $2 \tbracket{\A\B\A} = \frac{\kappa_{\B}}{\tbracket{\B}} \tbracket{\A\B}^{2}$, $2 \tbracket{\B\A\B} = \frac{\kappa_{\A}}{\tbracket{\A}} \tbracket{\A\B}^{2}$, and $\tbracket{\A\B\B} = 2 \frac{\kappa_{\B}}{\tbracket{\B}} \tbracket{\B\B} \tbracket{\A\B}$ for strictly positive constants $\kappa_{\A}$ and $\kappa_{\B}$~\parencite{demirel2014momentclosure}. Then, $\tilde{H}\tof{\tbracket{\A\A}, \tbracket{\B\B}} = \tparenth{ \frac{\kappa_{\B}}{\tbracket{\B}} \tbracket{\A\B} - 2 \frac{\kappa_{\A}}{\tbracket{\A}} \tbracket{\A\A}, \frac{\kappa_{\A}}{\tbracket{\A}} \tbracket{\A\B} - 2 \frac{\kappa_{\B}}{\tbracket{\B}} \tbracket{\B\B} }$, $\theta_{*} = \frac{\frac{\kappa_{\B}}{\tbracket{\B}}}{\frac{\kappa_{\A}}{\tbracket{\A}} + \frac{\kappa_{\B}}{\tbracket{\B}}}$, and $p_{*} = 1 - \frac{1}{2 M} \tparenth{\frac{\tbracket{\A}}{\kappa_{\A}} + \frac{\tbracket{\B}}{\kappa_{\B}}}$ assuming that $\frac{\tbracket{\A}}{\kappa_{\A}} + \frac{\tbracket{\B}}{\kappa_{\B}} < 2 M$. The latter is usually satisfied when the network is not too sparse so that the random-graph assumption $\kappa_{\A}$, $\kappa_{\B} \approx 1$ holds true~\parencite{demirel2014momentclosure}. Furthermore,
$\det{ D_{\tparenth{\tbracket{\A\A}, \tbracket{\B\B}}} \tilde{H}\tof{\theta_{*} M, \tparenth{1-\theta_{*}} M} } = 2 \tparenth{\frac{\kappa_{\A}}{\tbracket{\A}} + \frac{\kappa_{\B}}{\tbracket{\B}}}^{2} > 0$,
$\trace{\begin{psmallmatrix} +1 & -1 \\ -1 & +1 \end{psmallmatrix} D_{\tparenth{\tbracket{\A\A}, \tbracket{\B\B}}} \tilde{H}\tof{\theta_{*} M, \tparenth{1-\theta_{*}} M} } = -2 \tparenth{\frac{\kappa_{\A}}{\tbracket{\A}} + \frac{\kappa_{\B}}{\tbracket{\B}}} < 0$,
$\trace{ D_{\tparenth{\tbracket{\A\A}, \tbracket{\B\B}}} \tilde{H}\tof{\theta_{*} M, \tparenth{1-\theta_{*}} M} } \times \, \trace{ \begin{psmallmatrix} +1 & -1 \\ -1 & +1 \end{psmallmatrix} D_{\tparenth{\tbracket{\A\A}, \tbracket{\B\B}}} \tilde{H}\tof{\theta_{*} M, \tparenth{1-\theta_{*}} M} } = 6 \tparenth{\frac{\kappa_{\A}}{\tbracket{\A}} + \frac{\kappa_{\B}}{\tbracket{\B}}}^{2} > 0$,
and so that since $\trace{ D_{\tparenth{\tbracket{\A\A}, \tbracket{\B\B}}} \tilde{H}\tof{\theta_{*} M, \tparenth{1-\theta_{*}} M} } = -3 \tparenth{\frac{\kappa_{\A}}{\tbracket{\A}} + \frac{\kappa_{\B}}{\tbracket{\B}}} < 0$, this closure indeed yields a supercritical transcritical bifurcation.

In summary, through Theorems~\ref{thm:sis-epidemic-model-order-1-closure}, \ref{thm:sis-epidemic-model-order-2-closure}, and \ref{thm:adaptive-voter-model-order-2-closure}, we now have rigorous and systematic evidence showing why certain closures for the SIS epidemic and adaptive voter model work well empirically and correctly reproduce a dynamical feature seen in direct large-scale network simulations.

\section{Proofs of Theorems~\ref{thms:sis-epidemic-model-closure} and \ref{thm:adaptive-voter-model-order-2-closure}}

In this section, we will provide the proofs for the main results of the previous section, namely Theorems~\ref{thm:sis-epidemic-model-order-1-closure}, \ref{thm:sis-epidemic-model-order-2-closure}, and \ref{thm:adaptive-voter-model-order-2-closure}, and we will start with a result that allows us to factorise rational closure relations that vanish on a hyperplane.

\begin{lemma}
	\label{lem:hyperplane-polynomial-factorisation}
	Suppose that $P$ is a polynomial in $d$ variables such that it vanishes on the hyperplane $\tset{x \in \reals^{d} \suchthat \tinner{a}{x} = M}$ for some $0 \neq a \in \reals^{d}$ and $M \in \reals$. Then $M - \tinner{a}{x}$ divides $P\tof{x}$ so that $P\tof{x} = \tparenth{M - \tinner{a}{x}} \tilde{P}\tof{x}$, where $\tilde{P}$ is again a polynomial in $d$ variables.
\end{lemma}
\begin{proof}
	We first consider the special case where $a = e_{1}$ and $M = 0$, i.e. $P$ vanishes whenever $x_{1} = 0$. Then, writing $P\tof{x} = \sum_{\alpha \in \naturals_{0}^{d}} c_{\alpha} x^{\alpha}$ with $c_{\alpha} = 0$ for all but finitely many $\alpha$, where in the conventional notation for multiindices $\alpha \in \naturals_{0}^{d}$ ($d \in \naturals$), $\xi^{\alpha} := \prod_{i = 1}^{d} \xi_{i}^{\alpha_{i}}$ for every vector $\xi \in \reals^{d}$, we necessarily have that $c_{\alpha} = 0$ for every $\alpha$ with $\alpha_{1} = 0$. Hence, $P\tof{x} = \sum_{\stack{\alpha \in \naturals_{0}^{d}}{\alpha_{1} > 0}} c_{\alpha} x^{\alpha} = x_{1} \sum_{\stack{\alpha \in \naturals_{0}^{d}}{\alpha_{1} > 0}} c_{\alpha} x^{\alpha - e_{1}} = x_{1} \tilde{P}\tof{x}$, where $\tilde{P}$ is again a polynomial.
	
	For the general case, let $A$ be a bijective linear operator such that $A e_{1} = a$, and consider the map $\phi: \reals^{d} \to \reals^{d}$ with $\phi\tof{x} = M e_{1} - \transpose{A}x$.
	Since $\phi$ is an affine linear map, $P \circ \inverse{\phi}$ is also a polynomial. Moreover, $\tinner{a}{\inverse{\phi}\tof{x}} = \tinner{\inverse{A}a}{M e_{1} - x} = M - x_{1}$ so that $\tinner{a}{\inverse{\phi}\tof{x}} = M$ whenever $x_{1} = 0$, and therefore also $\tparenth{P \circ \inverse{\phi}}\tof{x} = 0$ by assumption. Hence, from the above, we conclude that $\tparenth{P \circ \inverse{\phi}}\tof{x} = x_{1} \widetilde{\tparenth{P \circ \inverse{\phi}}}\tof{x}$ and $P\tof{x} = \tparenth{P \circ \inverse{\phi}}\tof{\phi\tof{x}} = \phi\tof{x}_{1} \widetilde{\tparenth{P \circ \inverse{\phi}}}\tof{\phi\tof{x}} = \tparenth{M - \tinner{a}{x}} \widetilde{\tparenth{P \circ \inverse{\phi}}}\tof{\phi\tof{x}}$.
\end{proof}

\subsection{Proof of Theorem~\ref{thms:sis-epidemic-model-closure}}

A powerful result to prove the existence of bifurcations in high-dimensional systems is the Crandall--Rabinowitz Theorem. Since the proof of Theorem~\ref{thm:sis-epidemic-model-order-2-closure} particularly relies on it, we recall it here for convenience.

\begin{theorem*}[Crandall--Rabinowitz~{\parencite[Theorem~I.5.1 \& Section~I.6]{kielhoefer2012bifurcation}}]
	\label{thm:crandall-rabinowitz}
	Let $X$ and $Y$ be Banach spaces, $U \subset X$ and $\Lambda \subset \reals$ be open, and $F: U \times \Lambda \to Y$ be twice continuously differentiable. Assume that $F\tof{0,\lambda} = 0$ for all $\lambda \in \Lambda$ and that $F\tof{\slot,\lambda_{*}}$ for some $\lambda_{*} \in \Lambda$ is a (non-linear) Fredholm operator with $\dim\kernel{D_{x}F\tof{0,\lambda_{*}}} = 1 = \codim\range{D_{x}F\tof{0,\lambda_{*}}}$. Furthermore, if $\kernel{D_{x}F\tof{0,\lambda_{*}}} = \linearspan\tset{\hat{v}_{0}}$ and $Y / \range{D_{x}F\tof{0,\lambda_{*}}} = \linearspan\tset{\hat{v}^{*}_{0}}$ with $\tnorm{\hat{v}_{0}} = 1 = \tnorm{\hat{v}_{0}^{*}}$, suppose that $\hat{v}'_{0}$ is the functional such that $\hat{v}'_{0}\tbracket{\hat{v}^{*}_{0}} = 1$ and $\hat{v}'_{0}\tbracket{y} = 0$ for all $y \in \range{D_{x}F\tof{0,\lambda_{*}}}$ and that $D^{2}_{x \lambda}F\tof{0,\lambda_{*}}\tbracket{\hat{v}_{0}} \not\in \range{D_{x}F\tof{0,\lambda_{*}}}$, i.e. $\hat{v}'_{0}\tbracket{D^{2}_{x \lambda}F\tof{0,\lambda_{*}}\tbracket{\hat{v}_{0}}} \neq 0$.

	Then there exists a non-trivial continuously differentiable curve through $\tparenth{0, \lambda_{*}}$,
	\begin{equation*}
		\tset{\tparenth{x\tof{s}, \lambda\tof{s}} \suchthat s \in \tball{\delta}{0}, \tparenth{x\tof{0}, \lambda\tof{0}} = \tparenth{0, \lambda_{*}}}
	\end{equation*}
	for some $\delta$ sufficiently small, such that $F \equiv 0$ on that curve. Moreover, all solutions $\tparenth{x, \lambda}$ of $F\tof{x, \lambda} = 0$ in a neighbourhood of $\tparenth{0, \lambda_{*}}$ are on either the non-trivial curve or the trivial curve $\tset{\tparenth{0, \lambda_{*} + s} \suchthat s \in \tball{\delta}{0}}$.

	More specifically, we have that
	\begin{equation*}
		x'\tof{0} = \hat{v}_{0} \quad \text{and} \quad \lambda'\tof{0} = -\frac{1}{2} \frac{\hat{v}_{0}'\bracket{D^{2}_{xx}F\tof{0,\lambda_{*}}\tbracket{\hat{v}_{0},\hat{v}_{0}}}}{\hat{v}_{0}'\bracket{D^{2}_{x \lambda}F\tof{0,\lambda_{*}}\tbracket{\hat{v}_{0}}}}
	\end{equation*}
	where $D_{x}F\tof{0,\lambda_{*}}: X / \ker{D_{x}F\tof{0,\lambda_{*}}} \to \range{D_{x}F\tof{0,\lambda_{*}}}$ is an isomorphism and thus invertible so that the term $\tparenth{\varid - P}\inverse{D_{x}F\tof{0,\lambda_{*}}}\tparenth{\varid - Q}$ is well-defined with the projections $P: X \to \ker{D_{x}F\tof{0,\lambda_{*}}}$ and $Q: Y \to Y / \range{D_{x}F\tof{0,\lambda_{*}}}$.
	
	In particular, the point $\tparenth{0,\lambda_{*}}$, where the non-trivial and the trivial solution curves intersect, is called the bifurcation point. Moreover, it corresponds to a transcritical bifurcation if $\lambda'\tof{0} \neq 0$.
\end{theorem*}

In the situation where the spaces are Hilbert spaces, the statement simplifies. More specifically, $X/\ker{D_{x}F\tof{0,\lambda_{*}}} \cong \orthogonal{\ker{D_{x}F\tof{0,\lambda_{*}}}}$ and $Y / \range{D_{x}F\tof{0,\lambda_{*}}} \cong \orthogonal{\range{D_{x}F\tof{0,\lambda_{*}}}}$ since both $\ker{D_{x}F\tof{0,\lambda_{*}}}$ and $\range{D_{x}F\tof{0,\lambda_{*}}}$ are closed subspaces, as $D_{x}F\tof{0,\lambda_{*}}$ is bounded and $D_{x}F\tof{0,\lambda_{*}}$ is Fredholm, respectively. Consequently, $P = \tprojector{\hat{v}_{0}}{\hat{v}_{0}}$, $Q = \tprojector{\hat{v}^{*}_{0}}{\hat{v}^{*}_{0}}$, and $\hat{v}'_{0} = \tinner{\hat{v}^{*}_{0}}{\slot}$ so that, in particular, $\hat{v}'_{0}\tbracket{Q \slot} = \tinner{\hat{v}_{0}^{*}}{\slot}$, and
\begin{equation*}
	\lambda'\tof{0} = -\frac{1}{2} \frac{\inner{\hat{v}_{0}^{*}}{D^{2}_{xx}F\tof{0,\lambda_{*}}\tbracket{\hat{v}_{0},\hat{v}_{0}}}}{\inner{\hat{v}_{0}^{*}}{D^{2}_{x \lambda}F\tof{0,\lambda_{*}}\tbracket{\hat{v}_{0}}}} \mathpunctuation{.}
\end{equation*}

In the following, we will only prove Theorem~\ref{thm:sis-epidemic-model-order-2-closure} since Theorem~\ref{thm:sis-epidemic-model-order-1-closure} can be proved in essentially the same way.

\begin{proof}[Proof of Theorem~\ref{thm:sis-epidemic-model-order-2-closure}]
	Assuming that the closure relation $\tparenth{\tbracket{\S\S\I}, \tbracket{\I\S\I}} = H\tof{\tbracket{\I}, \tbracket{\S\I}, \tbracket{\I\I}}$ is rational and that $H\tof{\tbracket{\I}, \tbracket{\S\I}, \tbracket{\I\I}} = 0$ whenever $\tbracket{\S\I} = 0$, due to Lemma~\ref{lem:hyperplane-polynomial-factorisation}, we have that $H\tof{\tbracket{\I}, \tbracket{\S\I}, \tbracket{\I\I}} = \tbracket{\S\I} \tilde{H}\tof{\tbracket{\I}, \tbracket{\S\I}, \tbracket{\I\I}}$, where $\tilde{H}\tof{\tbracket{\I}, \tbracket{\S\I}, \tbracket{\I\I}}$ is again rational.
	
	Applying the closure relation to system \eqref{eq:reduced-general-sis-epidemic-model}, we obtain the system
	\begin{equation}
		\dot{\Xi} = \underbrace{\begin{psmallmatrix} -1 & \rho & 0 \\ 0 & -\tparenth{1 + \rho} & 1 \\ 0 & 2 \rho & -2 \\ \end{psmallmatrix} \Xi + \rho \, \Xi_{2} \begin{psmallmatrix} 0 & 0 \\ 1 & -1 \\ 0 & 2 \\ \end{psmallmatrix} \tilde{H}\tof{\Xi}}_{=: F\tof{\Xi,\rho}} \quad \text{with $\Xi \equiv \begin{psmallmatrix} \tbracket{\I} \\ \tbracket{\S\I} \\ \tbracket{\I\I} \end{psmallmatrix}$.}
	\end{equation}
	
	Then $F\tof{0,\rho} = 0$, and for $\rho_{*} = \frac{1}{\tilde{H}_{1}\tof{0}}$ we have that $D_{\Xi}F\tof{0,\rho_{*}} = \begin{psmallmatrix} -1 & \frac{1}{\tilde{H}_{1}\tof{0}} & 0 \\ 0 & -\frac{1 + \tilde{H}_{2}\tof{0}}{\tilde{H}_{1}\tof{0}} & 1 \\ 0 & 2 \frac{1 + \tilde{H}_{2}\tof{0}}{\tilde{H}_{1}\tof{0}} & -2 \end{psmallmatrix}$ so that $\rank{D_{\Xi}F\tof{0,\rho_{*}}} = 2$ and $\dim\kernel{D_{\Xi}F\tof{0,\rho_{*}}} = 1 = \codim\range{D_{\Xi}F\tof{0,\lambda_{*}}}$. In particular,
	\begin{itemize}[label=--]
		\item $\kernel{D_{\Xi}F\tof{0,\rho_{*}}} = \linearspan\tset{\hat{v}_{0}} = \linearspan\set{\frac{1}{\sqrt{1 + \tilde{H}_{1}\tof{0}^{2} + \tparenth{1 + \tilde{H}_{2}\tof{0}}^{2}}} \begin{psmallmatrix} 1 \\ \tilde{H}_{1}\tof{0} \\ 1 + \tilde{H}_{2}\tof{0} \\ \end{psmallmatrix}}$ and
		\item $\range{D_{\Xi}F\tof{0, \rho_{*}}} = \linearspan\set{\begin{psmallmatrix} 1 \\ 0 \\ 0 \\ \end{psmallmatrix}, \frac{1}{\sqrt{5}} \begin{psmallmatrix} 0 \\ 1 \\ -2 \\ \end{psmallmatrix}}$ so that $\orthogonal{\range{D_{\Xi}F\tof{0,\rho_{*}}}} = \linearspan\tset{\hat{v}_{0}^{*}} := \linearspan\set{\frac{1}{\sqrt{5}} \begin{psmallmatrix} 0 \\ 2 \\ 1 \\ \end{psmallmatrix}}$.
	\end{itemize}
	
	Thus, $\tinner{\hat{v}_{0}^{*}}{D_{\rho \Xi}F\tof{0,\rho_{*}}\tbracket{\hat{v}_{0}}} = \frac{2 \tilde{H}_{1}\tof{0}^{2}}{\sqrt{5} \sqrt{1 + \tilde{H}_{1}\tof{0}^{2} + \tparenth{1 + \tilde{H}_{2}\tof{0}}^{2}}} > 0$ so that $D^{2}_{\rho \Xi}F\tof{0,\rho_{*}}\tbracket{\hat{v}_{0}} \not\in \range{D_{\Xi}F\tof{0, \rho_{*}}}$. Hence, by the Crandall--Rabinowitz Theorem, we conclude that there exists a non-trivial bifurcation curve with $\rho\tof{0} = \rho_{*}$ and
	\begin{equation}
		\rho'\tof{0} = -\frac{1}{2} \frac{\tinner{\hat{v}_{0}^{*}}{D^{2}_{\Xi\Xi}F\tof{0, \rho_{*}}\tbracket{\hat{v}_{0}, \hat{v}_{0}}}}{\tinner{\hat{v}_{0}^{*}}{D^{2}_{\rho \Xi}F\tof{0, \rho_{*}}\tbracket{\hat{v}_{0}}}} = - \frac{\partial_{1} \tilde{H}_{1}\tof{0} + \tilde{H}_{1}\tof{0} \partial_{2} \tilde{H}_{1}\tof{0} + \tparenth{1 + \tilde{H}_{2}\tof{0}} \partial_{3} \tilde{H}_{1}\tof{0}}{\tilde{H}_{1}\tof{0}^{2} \sqrt{1 + \tilde{H}_{1}\tof{0}^{2} + \tparenth{1 + \tilde{H}_{2}\tof{0}}^{2}}} \mathpunctuation{.}
	\end{equation}
	Moreover, this bifurcation is transcritical if $\rho'\tof{0} \neq 0$ which, in turn, is the case if and only if $\partial_{1} \tilde{H}_{1}\tof{0} + \tilde{H}_{1}\tof{0} \partial_{2} \tilde{H}_{1}\tof{0} + \tparenth{1 + \tilde{H}_{2}\tof{0}} \partial_{3} \tilde{H}_{1}\tof{0} \neq 0$.
	
	With regard to stability, we note that
	\begin{equation}
		\Re{\spectrum{D_{\Xi}F\tof{0,\rho}}} = \tset{-1, \frac{1}{2} \tparenth{\phi\tof{\rho} \pm \Re{\sqrt{\phi\tof{\rho}^{2} - 8 \tparenth{1 - \tilde{H}_{1}\tof{0} \rho}}}}}
	\end{equation}
	with $-\phi\tof{\rho} = 3 - \tparenth{\tilde{H}_{1}\tof{0} - \tparenth{1 + \tilde{H}_{2}\tof{0}}} \rho$.
	For $\rho > \rho_{*}$, $8 \tparenth{1 - \tilde{H}_{1}\tof{0} \rho} < 0$ and thus $\Re{\sqrt{\phi\tof{\rho}^{2} - 8 \tparenth{1 - \tilde{H}_{1}\tof{0} \rho}}} > \tabs{\phi\tof{\rho}}$ so that $0$ is always an unstable solution. Conversely, for $\rho < \rho_{*}$, $\Re{\sqrt{\phi\tof{\rho}^{2} - 8 \tparenth{1 - \tilde{H}_{1}\tof{0} \rho}}} < \tabs{\phi\tof{\rho}}$ so that $0$ is a stable solution if $\phi\tof{\rho_{*}} < 0$.
	
	Hence, assuming that $\phi\tof{\rho_{*}} = - \tparenth{2 + \frac{1 + \tilde{H}_{2}\tof{0}}{\tilde{H}_{1}\tof{0}}} < 0$ or $2 \tilde{H}_{1}\tof{0} + \tilde{H}_{2}\tof{0} + 1 > 0$, we finally conclude that the bifurcation is supercritical (subcritical) if $\rho'\tof{0} \overset{\tof{<}}{>} 0$ and therefore if $\partial_{1} \tilde{H}_{1}\tof{0} + \tilde{H}_{1}\tof{0} \partial_{2} \tilde{H}_{1}\tof{0} + \tparenth{1 + \tilde{H}_{2}\tof{0}} \partial_{3} \tilde{H}_{1}\tof{0} \overset{\tof{>}}{<} 0$, which then completes the proof.
\end{proof}

\subsection{Proof of Theorem~\ref{thm:adaptive-voter-model-order-2-closure}}

Rather than a standard transcritical bifurcation, the adaptive voter model exhibits a degenerate transcritical bifurcation. Whereas in the former, upon parameter variation, two equilibrium points collide, the latter shall refer to the situation where there exists a whole manifold of codimension $1$ of equilibria, as well as an equilibrium point that collide upon parameter variation~(Fig.~\ref{fig:degenerate-transcritical-bifurcation}).

\begin{figure}[tbp]
	\centering
	
	\includegraphics[page=3,trim={0 0.45cm 0 0},clip]{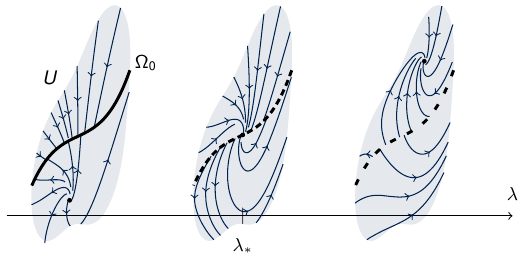}
	
	\caption{\textbf{Schematic of a degenerate transcritical bifurcation.} In an open neighbourhood $U$, there exists a manifold of codimension $1$ of equilibria, $\Omega_{0}$, as well as an equilibrium point. Upon variation of the bifurcation parameter $\lambda$, the equilibrium point eventually collides with and crosses the manifold of equilibria. At the bifurcation, when $\lambda = \lambda_{*}$ and the bifurcation lies in the manifold, there occurs an exchange of stability.}
	\label{fig:degenerate-transcritical-bifurcation}
\end{figure}

A prototypical example of a dynamical system for modelling such bifurcations is
\begin{equation}
	\dot{x} = F\tof{x, \lambda} := \omega\tof{x} f\tof{x, \lambda}
	\label{eq:degenerate-bifurcation-prototype-dynamical-system}
\end{equation}
with $\omega: \reals^{d} \to \reals$ and $f: \reals^{d} \times \reals \to \reals^{d}$. This system has the zeros of $\omega$ as well as $f$ as equilibria, and under further conditions their interplay can indeed give rise to bifurcations as described above.

Before getting to that, we state a few auxiliary results that we will need in the following.

\begin{lemma}
	\label{lem:rank-1-projection-spectrum}
	Let $X$ be a Hilbert space, and consider the rank-$1$ projection $\tprojector{x'}{x}$ for $x$, $x' \in X$. Then, if $x \neq 0$, $x$ is an eigenvector with corresponding eigenvalue $\tinner{x'}{x}$. All the other eigenvalues are $0$ with eigenvectors in $\orthogonal{\linearspan\tset{x'}}$.
\end{lemma}

\begin{lemma}
	\label{lem:degenerate-equilibrium-stability}
	Let $f: \reals^{d} \to \reals^{d}$, and consider the dynamical system $\dot{x} = F\tof{x} := x_{1} f\tof{x}$. Then every $x_{0} \in \tset{0} \times \reals^{d-1}$ for which $f_{1}\tof{x_{0}} \neq 0$ is a non-hyperbolic equilibrium, that is, stable if $f_{1}\tof{x_{0}} < 0$ and unstable otherwise.
\end{lemma}
\begin{proof}
	Without loss of generality, we may assume in the following that $x_{0} = 0$ with $f_{1}\tof{x_{0}} = f_{1}\tof{0} \neq 0$. Then, $DF\tof{x_{0}} = DF\tof{0} = \tprojector{e_{1}}{f\tof{0}}$ with eigenvalues $\tinner{e_{1}}{f\tof{0}} = f_{1}\tof{0}$ for the eigenvector $f\tof{0}$ and $0$ for the eigenvectors $e_{2}$, \ldots $e_{d}$ due to Lemma~\ref{lem:rank-1-projection-spectrum}. Hence, we have that $x_{0}$ is a non-hyperbolic equilibrium. Now, in order to conclude about its stability, we need to consider the dynamics on the $d-1$-dimensional center manifold.

	In order to perform the center-manifold reduction, define a linear operator $U$ together with its inverse by their action on the standard basis vectors via
	\begin{equation}
		\begin{split}
		&U : e_{i} \mapsto \begin{cases} e_{i+1} & \text{if $1 \leq i < d$} \\ f\tof{0} & \text{if $i = d$} \end{cases} \\
			&\quad \text{and} \quad \\
		&\inverse{U} : e_{i} \mapsto \begin{cases} \frac{1}{f_{1}\tof{0}} \parenth{e_{d} - \sum_{j = 1}^{d-1} f_{j+1}\tof{0} e_{j}} & \text{if $i = 1$} \\ e_{i-1} & \text{if $1 < i \leq d$} \end{cases} \mathpunctuation{.}
		\end{split}
	\end{equation}
	These operators are well-defined, and upon the transformation $x \mapsto \inverse{U} x$ the dynamical system is brought into a standard form with its linearisation in block-diagonal form. Indeed, for the system
	\begin{equation}
		\dot{x} = \hat{F}\tof{x} := \inverse{U} F\tof{U x} = f_{1}\tof{0} x_{d} \inverse{U} f\tof{U x}
		\label{eq:normal-form-dynamical-system-equation}
	\end{equation}
	we have that $D\hat{F}\tof{0} = f_{1}\tof{0} \tprojector{e_{d}}{e_{d}}$.

	If $h: \reals^{d-1} \to \reals$ is such that its graph represents the center manifold locally in a neighbourhood of $0$, then the invariance condition reads
	\begin{equation}
		\begin{split}
			0 &= Dh\tof{\xi}\bracket{\hat{F}_{1, \ldots d-1}\tof{\begin{smallmatrix}\xi \\ h\tof{\xi}\end{smallmatrix}}} - \hat{F}_{d}\tof{\begin{smallmatrix}\xi \\ h\tof{\xi}\end{smallmatrix}} \\
			&= h\tof{\xi} \left(\sum_{i = 1}^{d-1} Dh_{i}\tof{\xi} \left\lbrace f_{1}\tof{0} f_{i+1}\tof{U\tof{\begin{smallmatrix}\xi \\ h\tof{\xi}\end{smallmatrix}}} - f_{i+1}\tof{0} f_{1}\tof{U\tof{\begin{smallmatrix}\xi \\ h\tof{\xi}\end{smallmatrix}}} \right\rbrace - f_{1}\tof{U\tof{\begin{smallmatrix}\xi \\ h\tof{\xi}\end{smallmatrix}}}\right)
		\end{split}
		\label{eq:center-manifold-invariance-equation}
	\end{equation}
	subject to the condition that $h\tof{0} = 0$ and $Dh\tof{0} = 0$.
	In order to compute $h$ to arbitrary precision, we make a power-series ansatz,
	\begin{equation}
		h\tof{\xi} := \sum_{n = 0}^{\infty} \sum_{\stack{\alpha \in \naturals_{0}^{d-1}}{\tabs{\alpha} = n}} c_{\alpha} \xi^{\alpha}
	\end{equation}
	where we use the conventional notation for multiindices: For a multiindex $\alpha \in \naturals_{0}^{d}$ ($d \in \naturals$), we set $\tabs{\alpha} := \sum_{i = 1}^{d} \alpha_{i}$. Moreover, for every vector $\xi \in \reals^{d}$, $\xi^{\alpha} := \prod_{i = 1}^{d} \xi_{i}^{\alpha_{i}}$.
	
	To ensure that $h\tof{0} = 0$ and $Dh\tof{0} = 0$, we require that $c_{\alpha} = 0$ whenever $\tabs{\alpha} \leq 1$. However, we will show that, in fact, $c_{\alpha} = 0$ for every $\alpha$ if $h$ satisfies \eqref{eq:center-manifold-invariance-equation} so that $h \equiv 0$.
	To see that, suppose $c_{\alpha} = 0$ whenever $\tabs{\alpha} \leq N$ for some $N \in \naturals$. Then, if $h$ is to satisfy \eqref{eq:center-manifold-invariance-equation}, also $c_{\alpha} = 0$ whenever $\tabs{\alpha} \leq N+1$.

	Indeed, if $c_{\alpha} = 0$ for $\tabs{\alpha} \leq N$, we have that $h\tof{\xi} = \O\tof{\tnorm{\xi}^{N+1}}$, while $Dh\tof{\xi} = \O\tof{\tnorm{\xi}^{N}}$. Hence, the invariance condition \eqref{eq:center-manifold-invariance-equation} yields that
	\begin{equation}
		\begin{split}
			0 &= \parenth{\sum_{\stack{\alpha \in \naturals_{0}^{d-1}}{\tabs{\alpha} = N+1}} c_{\alpha} \xi^{\alpha} + \O\tof{\tnorm{\xi}^{N+2}}} \parenth{\O\tof{\tnorm{\xi}^{N}} - f_{1}\tof{0} + \O\tof{\tnorm{\xi}^{1}}} = - f_{1}\tof{0} \sum_{\stack{\alpha \in \naturals_{0}^{d-1}}{\tabs{\alpha} = N+1}} c_{\alpha} \xi^{\alpha} + \O\tof{\tnorm{\xi}^{N + 2}} \mathpunctuation{.}
		\end{split}
	\end{equation}
	For this to hold for every $\xi$, we require that $c_{\alpha} = 0$ for $\tabs{\alpha} = N+1$ so that, altogether, $c_{\alpha} = 0$ for $\tabs{\alpha} \leq N+1$.

	Hence, since by assumption $c_{\alpha} = 0$ whenever $\tabs{\alpha} \leq 1$, it follows inductively that $c_{\alpha} = 0$ for $\tabs{\alpha} \leq N$ for every $N \in \naturals$ so that $h\tof{\xi} = 0$.
	Thus, on the center manifold,
	\begin{equation}
		\dot{\xi} = \hat{F}_{1, \ldots d-1}\tof{\begin{smallmatrix}\xi \\ h\tof{\xi}\end{smallmatrix}} = 0
	\end{equation}
	meaning that in the center directions there are no dynamics. Hence, on the center manifold, $0$ is a stable equilibrium, and for the full system we conclude that stability of the origin depends only on the sign of $f_{1}\tof{0}$~\parencites[Theorem~2]{carr1981applications}[Theorem~10.6]{wiggins2017ordinary}.
\end{proof}

With that, we can now state the conditions for the existence of a bifurcation in the prototypical dynamical system \eqref{eq:degenerate-bifurcation-prototype-dynamical-system}.

\begin{proposition}
	\label{prop:degenerate-transcritical-bifurcation}
	Let $U \subset \reals^{d}$ and $\Lambda \subset \reals$ be open, and let $F = \omega \cdot f : U \times \Lambda \to \reals^{d}$ with $\omega: U \to \reals$ and $f: U \times \Lambda \to \reals^{d}$ at least once continuously differentiable. Assume that $\omega\tof{x_{0}} = 0$ for some $x_{0} \in U$ with $D\omega\tof{x_{0}} \neq 0$, and let $\Omega_{0} = \inverse{\omega}\tof{0}$. Assume furthermore that $f\tof{x_{0}, \lambda_{*}} = 0$ for some $\lambda_{*} \in \Lambda$, and suppose that $D_{x}f\tof{x_{0},\lambda_{*}}$ is invertible and that $D_{\lambda}f\tof{x_{0}, \lambda_{*}}$ and $\inverse{D_{x}f\tof{x_{0},\lambda_{*}}}\tbracket{D_{\lambda}f\tof{x_{0},\lambda_{*}}} \not\in T_{x_{0}}\Omega_{0}$.
	
	Then, the system
	\begin{equation*}
		\dot{x} = F\tof{x, \lambda} = \omega\tof{x} f\tof{x, \lambda}
	\end{equation*}
	exhibits a (degenerate) transcritical bifurcation at $x_{0}$ when $\lambda = \lambda_{*}$. More specifically, there exist open neighbourhoods $\tilde{U} \subset U$ and $\tilde{\Lambda} \subset \Lambda$ around $x_{0}$ and $\lambda_{*}$, respectively, so that $\Omega_{0} \cap \tilde{U}$ is a differential manifold of codimension $1$, and $F\tof{x, \lambda} = 0$ for every $x \in \Omega_{0} \cap \tilde{U}$ and $\lambda \in \tilde{\Lambda}$. In addition, there exist points $x\tof{\lambda}$ for every $\lambda \in \tilde{\Lambda}$ so that $F\tof{x\tof{\lambda}, \lambda} = 0$ with $x\tof{\lambda_{*}} = x_{0}$. In fact, if $F\tof{x, \lambda} = 0$ for some $x \in \tilde{U}$ and $\lambda \in \tilde{\Lambda}$, then $x \in \Omega_{0}$ or $x = x\tof{\lambda}$.
	
	Assuming further that $\Re\spectrum{D_{x}f\tof{x_{0},\lambda_{*}}} \subset \interval[open]{0}{\infty}$ or $\Re\spectrum{D_{x}f\tof{x_{0},\lambda_{*}}} \subset \interval[open]{-\infty}{0}$ and that
	\begin{itemize}[label=--]
		\item $\tinner{D_{\lambda}f\tof{x_{0}, \lambda_{*}}}{D\omega\tof{x_{0}}} \tinner{D\omega\tof{x_{0}}}{\inverse{D_{x}f\tof{x_{0}, \lambda_{*}}}\tbracket{D_{\lambda}f\tof{x_{0}, \lambda_{*}}}} > 0$ if $\Re\spectrum{D_{x}f\tof{x_{0},\lambda_{*}}} \subset \interval[open]{0}{\infty}$ and
		\item $\tinner{D_{\lambda}f\tof{x_{0}, \lambda_{*}}}{D\omega\tof{x_{0}}} \tinner{D\omega\tof{x_{0}}}{\inverse{D_{x}f\tof{x_{0}, \lambda_{*}}}\tbracket{D_{\lambda}f\tof{x_{0}, \lambda_{*}}}} < 0$ if $\Re\spectrum{D_{x}f\tof{x_{0},\lambda_{*}}} \subset \interval[open]{-\infty}{0}$,
	\end{itemize}
	dynamically, there is an exchange of stability along the curve $x\tof{\lambda}$ at $\lambda = \lambda_{*}$, so that if $x\tof{\lambda}$ is unstable for $\lambda < \lambda_{*}$, it is stable for $\lambda > \lambda_{*}$ and vice versa.
	
	More specifically, if $\Re\spectrum{D_{x}f\tof{x_{0},\lambda_{*}}} \subset \interval[open]{0}{\infty}$ and that curve crosses $\Omega_{0}$ in the direction of the gradient, the stability switches from stable to unstable, and if that curve crosses $\Omega_{0}$ in the opposite direction, the stability switches from unstable to stable. Conversely, if $\Re\spectrum{D_{x}f\tof{x_{0},\lambda_{*}}} \subset \interval[open]{-\infty}{0}$ and that curve crosses $\Omega_{0}$ in the direction of the gradient, the stability switches from unstable to stable and if that curve crosses $\Omega_{0}$ in the opposite direction, from stable to unstable.
\end{proposition}
\begin{proof}
	Without loss of generality, and upon choosing $U$ sufficiently small, we may assume that $D\omega\tof{x} \neq 0$ for every $x \in U$. Otherwise, since $D\omega\tof{x_{0}} \neq 0$, there exists an open neighbourhood $\tilde{U} \subset U$ around $x_{0}$ such that $D\omega\tof{x} \neq 0$ for every $x \in \tilde{U}$.
	Thus, $\omega : U \to \reals$ is a differentiable map of constant rank, so that $\Omega_{0} = \inverse{\omega}\tof{0}$ is a differentiable manifold of codimension $1$ in the ambient space~\parencites[Theorem~2.5.3]{conlon2001differentiable}[Corollary~5.14]{lee2012introduction}. Importantly, by construction, $F\tof{x, \lambda} = 0$ for every $x \in \Omega_{0}$ and $\lambda \in \Lambda$ so that $\Omega_{0}$ constitutes a manifold of trivial equilibria.
	
	Since $D_{x}F\tof{x, \lambda} = f\tof{x, \lambda} D\omega\tof{x}$ for every $x \in \Omega_{0}$ and $D\omega\tof{x} \perp T_{x}\Omega_{0}$, where $T_{x}\Omega_{0}$ is the tangent space to $\Omega_{0}$ at $x$, we have that $T_{x}\Omega_{0} \subseteq \kernel{D_{x}F\tof{x,\lambda}}$ and $\dim\kernel{D_{x}F\tof{x,\lambda}} \geq d - 1$. In particular, if $f\tof{x,\lambda} \neq 0$, $D\omega\tof{x} \not\in \kernel{D_{x}F\tof{x,\lambda}}$, and thus, in fact, $\dim\kernel{D_{x}F\tof{x,\lambda}} = d - 1$.
	Hence, since there exists $\lambda_{*} \in \Lambda$ such that $f\tof{x_{0},\lambda_{*}} = 0$, we have that $D_{x}F\tof{x_{0}, \lambda}$ acquires an addition defect when $\lambda = \lambda_{*}$.
	
	Now, since $D_{x}f\tof{x_{0},\lambda_{*}}$ is invertible, by the Implicit Function Theorem~\parencite[Theorem~I.1.1]{kielhoefer2012bifurcation}, there exists $\phi : \Lambda \to U$ such that $\phi\tof{\lambda_{*}} = x_{0}$ and $D\phi\tof{\lambda_{*}} = - \inverse{D_{x}f\tof{x_{0},\lambda_{*}}}\tbracket{D_{\lambda}f\tof{x_{0},\lambda_{*}}}$. Moreover, if $f\tof{x,\lambda} = 0$ on $U \times \Lambda$, then necessarily $x \in \range{\phi}$, where, as before and to simplify notation, we assume that $U$ and $\Lambda$ have been chosen sufficiently small, so that it is not necessary to pass to small neighbourhoods $\tilde{U}$ and $\tilde{\Lambda}$.
	This establishes the existence of a branch of non-trivial equilibria in the extended phase space in addition to the equilibria on the manifold $\Omega_{0}$. In addition, this branch intersects the manifold non-tangentially at $x_{0}$ for $\lambda = \lambda_{*}$, since $\tinner{D\omega\tof{x_{0}}}{D\phi\tof{\lambda_{*}}} = - \tinner{D\omega\tof{x_{0}}}{\inverse{D_{x}f\tof{x_{0},\lambda_{*}}}\tbracket{D_{\lambda}f\tof{x_{0},\lambda_{*}}}} \neq 0$ by assumption.
	
	Regarding the dynamic stability of these equilibrium points, we have that $D_{x}F\tof{\phi\tof{\lambda}, \lambda} = \omega\tof{\phi\tof{\lambda}} D_{x}f\tof{\phi\tof{\lambda}, \lambda}$ so that $\spectrum{D_{x}F\tof{\phi\tof{\lambda}, \lambda}} = \omega\tof{\phi\tof{\lambda}} \, \spectrum{D_{x}f\tof{\phi\tof{\lambda}, \lambda}}$. By assumption, we have that either \linebreak $\Re\spectrum{D_{x}f\tof{x_{0},\lambda_{*}}} \subset \interval[open]{0}{\infty}$ or $\Re\spectrum{D_{x}f\tof{x_{0},\lambda_{*}}} \subset \interval[open]{-\infty}{0}$ so that consequently also $\Re\spectrum{D_{x}f\tof{\phi\tof{\lambda}, \lambda}} \subset \interval[open]{0}{\infty}$ or $\Re\spectrum{D_{x}f\tof{\phi\tof{\lambda}, \lambda}} \subset \interval[open]{-\infty}{0}$ for $\lambda$ in a neighbourhood of $\lambda_{*}$. Moreover, $\omega\tof{\phi\tof{\lambda}}$ changes its sign at $\lambda = \lambda_{*}$ since $\derivative[\lambda_{*}]{\lambda} \omega\tof{\phi\tof{\lambda}} = \tinner{D\omega\tof{x_{0}}}{D\phi\tof{\lambda_{*}}} \neq 0$.
	
	As for some $x \in \Omega_{0}$, owing to Lemma~\ref{lem:degenerate-equilibrium-stability}, we have that it is a non-hyperbolic equilibrium point whose stability is determined by the only non-vanishing eigenvalue $\tinner{D\omega\tof{x}}{f\tof{x,\lambda}}$. As before, this changes its sign at $\lambda = \lambda_{*}$ since $\derivative[\lambda_{*}]{\lambda} \tinner{D\omega\tof{x}}{f\tof{x,\lambda}} = \tinner{D\omega\tof{x}}{D_{\lambda}f\tof{x,\lambda_{*}}} \neq 0$. This is because, by assumption, $D_{\lambda}f\tof{x_{0},\lambda_{*}} \not\in T_{0}\Omega_{0}$ and therefore $\tinner{D\omega\tof{x_{0}}}{D_{\lambda}f\tof{x_{0}, \lambda_{*}}} \neq 0$, so that also $\tinner{D\omega\tof{x}}{D_{\lambda}f\tof{x,\lambda_{*}}} \neq 0$ in some sufficiently small neighbourhood of in $x_{0}$ in $\Omega_{0}$.
	
	Hence, finally, if $\Re\spectrum{D_{x}f\tof{x_{0},\lambda_{*}}} \subset \interval[open]{0}{\infty}$ and
	\begin{equation}
		\sign{}\tinner{D\omega\tof{x_{0}}}{D_{\lambda}f\tof{x_{0},\lambda_{*}}} = +1 = \sign\tinner{D\omega\tof{x_{0}}}{\linebreak \inverse{D_{x}f\tof{x_{0},\lambda_{*}}}\tbracket{D_{\lambda}f\tof{x_{0},\lambda_{*}}}} \mathpunctuation{,}
	\end{equation}
	$\derivative[\lambda_{*}]{\lambda} \omega\tof{\phi\tof{\lambda}} < 0$ so that the non-trivial equilibrium switches from being stable to being unstable.
	At the same time, $\derivative[\lambda_{*}]{\lambda} \tinner{D\omega\tof{x}}{f\tof{x,\lambda}} > 0$ so that the trivial equilibria in a neighbourhood of $x_{0}$ switch from being stable to being unstable.
	Similarly, if
	\begin{equation}
		\sign{\tinner{D\omega\tof{x_{0}}}{D_{\lambda}f\tof{x_{0},\lambda_{*}}}} = -1 = \sign\tinner{D\omega\tof{x_{0}}}{\inverse{D_{x}f\tof{x_{0},\lambda_{*}}}\tbracket{D_{\lambda}f\tof{x_{0},\lambda_{*}}}} \mathpunctuation{,}
	\end{equation}
	we obtain the reverse situation.
	Furthermore, if conversely $\Re\spectrum{D_{x}f\tof{x_{0},\lambda_{*}}} \subset \interval[open]{-\infty}{0}$ and
	\begin{equation}
		\sign{}\tinner{D\omega\tof{x_{0}}}{D_{\lambda}f\tof{x_{0},\lambda_{*}}} = +1 \neq -1 = \sign{}\tinner{D\omega\tof{x_{0}}}{\inverse{D_{x}f\tof{x_{0},\lambda_{*}}}\tbracket{D_{\lambda}f\tof{x_{0},\lambda_{*}}}} \mathpunctuation{,}
	\end{equation}
	we have that $\derivative[\lambda_{*}]{\lambda} \omega\tof{\phi\tof{\lambda}} < 0$, so that the non-trivial equilibrium switches from being unstable to being stable. Again, at the same time, $\derivative[\lambda_{*}]{\lambda} \tinner{D\omega\tof{x}}{f\tof{x,\lambda}} > 0$ so that the trivial equilibria in a neighbourhood of $x_{0}$ switch from being stable to being unstable. Similarly, if
	\begin{equation}
		\sign{}\tinner{D\omega\tof{x_{0}}}{D_{\lambda}f\tof{x_{0},\lambda_{*}}} = -1 \neq +1 = \sign{}\tinner{D\omega\tof{x_{0}}}{\inverse{D_{x}f\tof{x_{0},\lambda_{*}}}\tbracket{D_{\lambda}f\tof{x_{0},\lambda_{*}}}} \mathpunctuation{,}
	\end{equation}
	we obtain the reverse situation.
	
	In summary, we have that if $\Re\spectrum{D_{x}f\tof{x_{0},\lambda_{*}}} \subset \interval[open]{0}{\infty}$ and the branch of non-trivial equilibria crosses the manifold of trivial equilibria in the direction of its gradient, the non-trivial equilibrium switches from being stable to being unstable, while in the opposite direction it switches from being unstable to being stable. Conversely, if $\Re\spectrum{D_{x}f\tof{x_{0},\lambda_{*}}} \subset \interval[open]{-\infty}{0}$ and the branch of non-trivial equilibria crosses the manifold of trivial equilibria in the direction of the manifold's gradient, the non-trivial equilibrium switches from being unstable to being stable and, in the opposite direction, from being stable to being unstable.
\end{proof}

For the proof of Theorem~\ref{thm:adaptive-voter-model-order-2-closure}, we only need the planar case ($d = 2$). Here, the spectral conditions simplify considerably as the localisation of the spectrum in either the left or right half-plane can be characterised in terms of the determinant and trace.

\begin{lemma}
	\label{lem:2x2-spectral-localisation}
	Let $A \in \operatorname{Mat}_{\reals}\tof{2 \times 2}$. Then, $\Re\spectrum{A} \subset \interval[open]{-\infty}{0}$ if and only if $\det{A} > 0$ and $\trace{A} < 0$. Conversely, $\Re\spectrum{A} \subset \interval[open]{0}{\infty}$ if and only if $\det{A} > 0$ and $\trace{A} > 0$.
\end{lemma}
\begin{proof}
	Indeed, if $A \in \operatorname{Mat}_{\reals}\tof{2 \times 2}$, then $2 \Re{\lambda_{\pm}} = \trace{A} \pm \Re\sqrt{\trace{A}^{2} - 4 \det{A}}$. The latter has a definite sign if and only if $\Re\sqrt{\trace{A}^{2} - 4 \det{A}} < \tabs{\trace{A}}$, and thus $\det{A} > 0$. Furthermore, in that case $\Re{\lambda_{\pm}} > 0$ or $\Re{\lambda_{\pm}} < 0$ if and only if $\trace{A} > 0$ or $\trace{A} < 0$, respectively.
\end{proof}

With that we obtain the following as a corollary of Proposition~\ref{prop:degenerate-transcritical-bifurcation}.

\begin{corollary}
	\label{cor:planar-degenerate-transcritical-bifurcation}
	Let $U \subset \reals^{2}$ and $\Lambda \subset \reals$ be open, and let $F = \omega \cdot f : U \times \Lambda \to \reals^{2}$ with $\omega: U \to \reals$ and $f: U \times \Lambda \to \reals^{2}$ at least once continuously differentiable. Assume that $\omega\tof{x_{0}} = 0$ for some $x_{0} \in U$ with $D\omega\tof{x_{0}} \neq 0$, and let $\Omega_{0} = \inverse{\omega}\tof{0}$. Assume furthermore that $f\tof{x_{0}, \lambda_{*}} = 0$ for some $\lambda_{*} \in \Lambda$, and suppose that $D_{x}f\tof{x_{0},\lambda_{*}}$ is invertible and that $D_{\lambda}f\tof{x_{0}, \lambda_{*}}$ and $\inverse{D_{x}f\tof{x_{0},\lambda_{*}}}\tbracket{D_{\lambda}f\tof{x_{0},\lambda_{*}}} \not\in T_{x_{0}}\Omega_{0}$.
	
	Then, the system
	\begin{equation*}
		\dot{x} = F\tof{x, \lambda} = \omega\tof{x} f\tof{x, \lambda}
	\end{equation*}
	exhibits a (degenerate) transcritical bifurcation at $x_{0}$ when $\lambda = \lambda_{*}$. More specifically, there exists open neighbourhoods $\tilde{U} \subset U$ and $\tilde{\Lambda} \subset \Lambda$ around $x_{0}$ and $\lambda_{*}$, respectively, so that $\Omega_{0} \cap \tilde{U}$ is a differential manifold of codimension $1$, and $F\tof{x, \lambda} = 0$ for every $x \in \Omega_{0} \cap \tilde{U}$ and $\lambda \in \tilde{\Lambda}$. In addition, there exist points $x\tof{\lambda}$ for every $\lambda \in \tilde{\Lambda}$ so that $F\tof{x\tof{\lambda}, \lambda} = 0$ with $x\tof{\lambda_{*}} = x_{0}$. In fact, if $F\tof{x, \lambda} = 0$ for some $x \in \tilde{U}$ and $\lambda \in \tilde{\Lambda}$, then $x \in \Omega_{0}$ or $x = x\tof{\lambda}$.
	
	Assuming further that $\det{D_{x}f\tof{x_{0},\lambda_{*}}} > 0$ and that
	\begin{equation*}
		\tinner{D_{\lambda}f\tof{x_{0}, \lambda_{*}}}{D\omega\tof{x_{0}}} \tinner{D\omega\tof{x_{0}}}{\inverse{D_{x}f\tof{x_{0}, \lambda_{*}}}\tbracket{D_{\lambda}f\tof{x_{0}, \lambda_{*}}}} \trace{D_{x}f\tof{x_{0},\lambda_{*}}} > 0 \mathpunctuation{,}
	\end{equation*}
	dynamically, there is a switch in stability along the curve $x\tof{\lambda}$ at $\lambda = \lambda_{*}$, so that if $x\tof{\lambda}$ is unstable for $\lambda < \lambda_{*}$ it is stable for $\lambda > \lambda_{*}$ and vice versa.
	
	More specifically, if $\trace{D_{x}f\tof{x_{0},\lambda_{*}}} > 0$ and that curve crosses $\Omega_{0}$ in the direction of the gradient, the stability switches from stable to unstable, and if that curve crosses $\Omega_{0}$ in the opposite direction, it switches from unstable to stable. Conversely, if $\trace{D_{x}f\tof{x_{0},\lambda_{*}}} < 0$ and that curve crosses $\Omega_{0}$ in the direction of the gradient, the stability switches from unstable to stable, and if that curve crosses $\Omega_{0}$ in the opposite direction, the stability switches from stable to unstable.
\end{corollary}
\begin{proof}
	
	Applying Lemma~\ref{lem:2x2-spectral-localisation} to $D_{x}f\tof{x_{0}, \lambda_{*}}$ and noting that
	\begin{equation}
		\sign{\tinner{D_{\lambda}f\tof{x_{0}, \lambda_{*}}}{D\omega\tof{x_{0}}} \tinner{D\omega\tof{x_{0}}}{\inverse{D_{x}f\tof{x_{0}, \lambda_{*}}}\tbracket{D_{\lambda}f\tof{x_{0}, \lambda_{*}}}}} = \sign{\trace{D_{x}f\tof{x_{0},\lambda_{*}}}}
	\end{equation}
	holds if and only if
	\begin{equation}
		\tinner{D_{\lambda}f\tof{x_{0}, \lambda_{*}}}{D\omega\tof{x_{0}}} \tinner{D\omega\tof{x_{0}}}{\inverse{D_{x}f\tof{x_{0}, \lambda_{*}}}\tbracket{D_{\lambda}f\tof{x_{0}, \lambda_{*}}}} \trace{D_{x}f\tof{x_{0},\lambda_{*}}} > 0
	\end{equation}
	the result follows from Proposition~\ref{prop:degenerate-transcritical-bifurcation}.
\end{proof}

With this results we are now finally in position to prove Theorem~\ref{thm:adaptive-voter-model-order-2-closure}.

\begin{proof}[Proof of Theorem~\ref{thm:adaptive-voter-model-order-2-closure}]
	Assuming that the closure relation $\tparenth{2\,\tbracket{\A\B\A} - \tbracket{\A\A\B}, 2\,\tbracket{\B\A\B} - \tbracket{\A\B\B}} = H\tof{\tbracket{\A\A}, \tbracket{\B\B}}$ is rational and that $H\tof{\tbracket{\A\A}, \tbracket{\B\B}} = 0$ whenever $\tbracket{\A\A} + \tbracket{\B\B} = M$, due to Lemma~\ref{lem:hyperplane-polynomial-factorisation} we have that $H\tof{\tbracket{\A\A}, \tbracket{\B\B}} = \tparenth{M - \tbracket{\A\A} - \tbracket{\B\B}} \tilde{H}\tof{\tbracket{\A\A}, \tbracket{\B\B}}$, where $\tilde{H}\tof{\tbracket{\A\A}, \tbracket{\B\B}}$ is again rational.

	Applying the closure relation to system \eqref{eq:reduced-general-adaptive-voter-model}, we obtain the system
	
	\begin{equation}
		\dot{\Xi} = \underbrace{\frac{1}{2} \tparenth{M - \Xi_{1} - \Xi_{2}}}_{=: \omega\tof{\Xi}} \underbrace{\parenth{ \begin{psmallmatrix} 1 \\ 1 \end{psmallmatrix} + \tparenth{1 - p} \tilde{H}\tof{\Xi} }}_{=: f\tof{\Xi, p}} \quad \text{with $\Xi \equiv \begin{psmallmatrix} \tbracket{\A\A} \\ \tbracket{\B\B} \end{psmallmatrix}$.}
	\end{equation}
	
	By assumption, there exists $\Xi_{0} = \begin{psmallmatrix} \theta_{*} M \\ \tparenth{1 - \theta_{*}} M \end{psmallmatrix}$ for which $\omega\tof{\Xi_{0}} = 0$ and $D\omega\tof{\Xi_{0}} = -\frac{1}{2} \begin{psmallmatrix} 1 \\ 1 \end{psmallmatrix} \neq 0$. Moreover, $\tilde{H}_{1}\tof{\Xi_{0}} = \tilde{H}_{2}\tof{\Xi_{0}} < -1$ so that $p_{*} = 1 + \frac{1}{\tilde{H}_{1}\tof{\Xi_{0}}} = 1 + \frac{1}{\tilde{H}_{2}\tof{\Xi_{0}}}$ with $0 < p_{*} < 1$ and $f\tof{\Xi_{0}, p_{*}} = 0$.
	
	Since $D_{\Xi}f\tof{\Xi_{0}, p_{*}} = \tparenth{1 - p_{*}} D_{\Xi}\tilde{H}\tof{\Xi_{0}}$, $D_{\Xi}f\tof{\Xi_{0}, p_{*}}$ is invertible as $\det{D_{\Xi}\tilde{H}\tof{\Xi_{0}}} \neq 0$.
	
	Now, using the fact that for any $A \in \operatorname{Mat}_{\reals}\tof{2 \times 2}$,
	\begin{equation}
		\inner{\begin{psmallmatrix} 1 \\ 1 \end{psmallmatrix}}{\inverse{A} \begin{psmallmatrix} 1 \\ 1 \end{psmallmatrix}} = \frac{\trace{\begin{psmallmatrix} 1 & -1 \\ -1 & 1 \end{psmallmatrix} A}}{\det{A}} \mathpunctuation{,}
	\end{equation}
	we get that $D_{p}f\tof{\Xi_{0}, p_{*}}$ and $\inverse{D_{\Xi}f\tof{\Xi_{0}, p_{*}}}\tbracket{D_{p}f\tof{\Xi_{0}, p_{*}}} \not\in T_{\Xi_{0}}\Omega_{0}$ since $D_{p}f\tof{\Xi_{0}, p_{*}} = - \tilde{H}\tof{\Xi_{0}} = \frac{1}{1 - p_{*}} \begin{psmallmatrix} 1 \\ 1 \end{psmallmatrix}$ and also $D\omega\tof{\Xi_{0}} \propto \begin{psmallmatrix} 1 \\ 1 \end{psmallmatrix}$.
	
	Thus, the claim follows from Corollary~\ref{cor:planar-degenerate-transcritical-bifurcation} by noting that furthermore $\det{D_{\Xi}f\tof{\Xi_{0}, p_{*}}} > 0$ if and only if $\det{D_{\Xi}\tilde{H}\tof{\Xi_{0}}} > 0$ and that in this case
	\begin{equation}
		\begin{split}
			&\mathrelphantom{=} \tinner{D_{p}f\tof{\Xi_{0}, p_{*}}}{D\omega\tof{\Xi_{0}}} \tinner{D\omega\tof{\Xi_{0}}}{\inverse{D_{\Xi}f\tof{\Xi_{0}, p_{*}}}\tbracket{D_{p}f\tof{\Xi_{0}, p_{*}}}} \trace{D_{\Xi}f\tof{\Xi_{0},p_{*}}} \\
			&= \frac{1}{2 \tparenth{1 - p_{*}}^{2}} \inner{\begin{psmallmatrix} 1 \\ 1 \end{psmallmatrix}}{\inverse{D_{\Xi}\tilde{H}\tof{\Xi_{0}}} \begin{psmallmatrix} 1 \\ 1 \end{psmallmatrix}} \trace{D_{\Xi}\tilde{H}\tof{\Xi_{0}}} \\
			&= \frac{1}{2 \tparenth{1 - p_{*}}^{2}} \frac{\trace{D_{\Xi}\tilde{H}\tof{\Xi_{0}}}}{\det{D_{\Xi}\tilde{H}\tof{\Xi_{0}}}} \trace{\begin{psmallmatrix} 1 & -1 \\ -1 & 1 \end{psmallmatrix} D_{\Xi}\tilde{H}\tof{\Xi_{0}}} > 0
		\end{split}
	\end{equation}
	if and only if $\trace{D_{\Xi}\tilde{H}\tof{\Xi_{0}}} \trace{\begin{psmallmatrix} 1 & -1 \\ -1 & 1 \end{psmallmatrix} D_{\Xi}\tilde{H}\tof{\Xi_{0}}} > 0$ provided that $\det{D_{\Xi}\tilde{H}\tof{\Xi_{0}}} > 0$.

	As for the criticality of the bifurcation, we note that in the direction of the gradient to the manifold of trivial equilibria, the number of $\tbracket{\A\A} + \tbracket{\B\B}$ is decreasing and, consequently, $\tbracket{\A\B}$ is increasing. Hence, the bifurcation is supercritical (subcritical) if the branch of non-trivial equilibria crosses the manifold of trivial equilibria either in the direction of the gradient while gaining (losing) stability or in the opposite direction while losing (gaining) stability.
	Thus, the bifurcation is supercritical (subcritical) if $\trace{D_{\Xi}\tilde{H}\tof{\Xi_{0}}} \overset{\tof{>}}{<} 0$.
\end{proof}

\section{Discussion}

In this work, we have studied the problem of finding moment closure relations from a dynamical perspective, focusing specifically on bifurcations. We considered the mean-field moment system of the SIS epidemic and the adaptive voter model and derived concrete conditions that a closure for these systems must satisfy in order for the closed mean-field system to locally exhibit the kind of transcritical bifurcation that is apparent in the mean of numerical simulations of the stochastic particle system and that is therefore expected. Even though their bifurcation is seemingly the same, the crucial difference between the two systems is that the trivial equilibrium is degenerate in the latter, making the analysis more involved because standard results from bifurcation theory cannot be directly applied.

In the light of the conditions for existence of a bifurcation, we then examined some frequently used closures for both the SIS epidemic and the adaptive voter mean-field model and showed that they indeed satisfy these conditions. Interestingly, in the case of the SIS epidemic model, we showed that the standard triple closure always gives rise to a transcritical bifurcation, but whether it is super- or subcritical depends on the coefficient~\parencite{kuehn2021universal}. Moreover, we also considered some hypothetical closures, which illustrated how our results can determine good classes of closure relations, provide clear evidence for their validity, and identify failures precisely.
More generally, the conditions that we have derived for a closure to preserve the local bifurcation immediately translate into a Taylor expansion of such a closure around the bifurcation point that characterises the set of good closures. The closures in this set are all equivalent in that they preserve the local bifurcation, so that in order to further constrain this set, one would need to formulate additional criteria.

To derive conditions on a moment closure relation that guarantee the existence of bifurcation, we \textit{a priori} assumed that a generic closure relation exists. Using this closure, we obtained a closed, low-dimensional system on which we performed a bifurcation analysis, which resulted in said conditions for bifurcations that were motivated by observations of numerical simulations. However, this is very different from performing a bifurcation analysis on the entire mean-field system. Results from the probabilistic analysis of contact processes do indeed suggest that for contact processed on an infinite network, a branching point, and thus a bifurcation, exists~\parencite[Chapter~VI]{liggett1999stochastic}, but showing this in the exact mean-field system is an open problem.

Importantly, in this work we have focused solely on transcritical bifurcations with the advantage that the location of the equilibrium in phase space does not depend on the closure. While the same would be true also for a pitchfork bifurcation from a trivial branch, it is not generically true, for instance, for a saddle-node bifurcation. In fact, for these other bifurcations, the dependence of the equilibrium location on the closure itself poses additional challenges when one is attempting to derive conditions that ensure the bifurcation is preserved.

It may be tempting to extend our approach across all of phase and parameter space at once; i.e. instead of requiring consistency within the main bifurcation points and localised dynamical phenomena one by one, one could aim at starting from global structures. Yet, imposing global constraints on the dynamics, e.g. as proposed in the context of fluid mechanics via the Jacobi identity for the Poisson brackets~\parencite{edwards1997time}, drastically limits the space of possible, potentially very well approximatingm and highly practical closures.

Overall, the idea of shifting the perspective from moment closures that provide quantitative good approximations across the entire phase space to moment closures that give rise to key qualitative features that one might want to preserve through the closure has potentially far-reaching consequences. While the former requires one to draw on heuristics, the latter allows one to produce a concrete classification of suitable closure relations, as we have demonstrated in this work.
Moreover, with the models already being an abstraction of real-world phenomena, in applications it is, arguably, often more important to have low-dimensional representations of a given system that reproduce qualitative features such as bifurcations than to require them to also be quantitatively accurate~\parencite{valdes2011built,scheffer2012anticipating}.

However, this does not have to be a rejection of moment closures that also provide quantitative approximations. On the contrary, we argue that if the latter is important or desirable, the conditions that one obtains when trying to preserve some qualitative features can, in fact, be used to guide the search for relations that also yield quantitative good approximations.
For instance, in recent years there has been a lot of research on learning the governing equations of non-linear dynamical systems from data~\parencite{schmidt2009distilling,brunton2016discovering,daniels2015automated,raissi2018deep}. Hence, if the dynamical system is known up to terms of certain order, these methods could also be utilised to learn closures and conditions like those we have derived can be implemented to inform and guide the learning procedure. As such, the approach we propose here can help to explain the validity of closures that may have been found either intuitively or by a data-driven approach, making it quite flexible.

Viewing moment closures in terms of certain qualitative features that one wants to preserve in the closed system is not restricted to bifurcations. In fact, to illustrate this point and give a simple example, we note that, instead of a whole bifurcation, one could just attempt to characterise closures that preserve one of the equilibria. In addition, as moment closure methods are applied more widely and not limited to network dynamical systems, it will be interesting to see whether the ideas presented here can be carried over to, e.g. stochastic differential equations and the moment systems associated with the solution process and some qualitative feature of the same~\parencite{kuehn2016moment}.

\subsection*{Acknowledgements}
J.M. and C.K. acknowledge funding from the Deut\-sche For\-schungs\-ge\-mein\-schaft (DFG, German Research Foundation) as well as partial support from the Volks\-wagen\-Stif\-tung (Volks\-wagen Foundation) via a Lich\-ten\-berg Professorship awarded to C.K.

\clearpage

\printbibliography[heading=bibintoc]

\end{document}